%% file: main.tex
\documentclass[11pt,letterpaper]{article}
\input{math_commands.tex}

\usepackage{hyperref}
\usepackage{cleveref}
\usepackage[notes=true,later=false,camera=false]{dtrt}

\title{Switching Graph Matrix Norm Bounds: \\  from i.i.d. to Random Regular Graphs}
\author{Jeff Xu \thanks{Carnegie Mellon University, \texttt{jeffxusichao@cmu.edu}. Supported in part by NSF CAREER Award \#2047933.}}
\date{\today}

\begin{document}

\maketitle

\begin{abstract}
In this work, we give novel spectral norm bounds for graph matrix on inputs being random regular graphs. Graph matrix is a family of random matrices with entries given by polynomial functions of the underlying input. These matrices have been known to be the backbone for the analysis of various average-case algorithms and hardness. Previous investigations of such matrices are largely restricted to the \Erdos-\Renyi model, and tight matrix norm bounds on regular graphs are only known for specific examples. We unite these two lines of investigations, and give the first result departing from the \Erdos-\Renyi setting in the full generality of graph matrices. We believe our norm bound result would enable a simple transfer of spectral analysis for average-case algorithms and hardness between these two distributions of random graphs.

As an application of our spectral norm bounds, we show that higher-degree Sum-of-Squares lower bounds for the independent set problem on \Erdos-\Renyi random graphs can be switched into lower bounds on random $d$-regular graphs. Our result is the first to address the general open question of analyzing higher-degree Sum-of-Squares on random regular graphs. \end{abstract}

\section{Introduction}
%
%
%

It is well-known in random graph theory that random $d$-regular graph denoted as $G_d(n)$ and \Erdos-\Renyi random graph $G(n, \frac{d}{n})$ share various phase transition thresholds at least in the leading order, eg. connectivity threshold, independence/chromatic number (see \cite{introtorandomgraph2015} for more examples) . For algorithmists,
one intriguing question to ask is whether its algorithmic variant holds:	 for average-case computational problems (of both algorithms and hardness), 
	   can the result and the analysis be transferred from one distribution to another under some general condition?
 Crucially, this question remains unclear in various situations as analysis of average-case problems (eg. \cite{HSS15, MW19, JP24} and see more examples in \cite{AMP20}), especially those appealing to a spectral reasoning, is usually brittle in the presence of correlation of the underlying input.
 
%
%

Motivated by the above question, we study the certification question of independent sets in sparse random graphs. In general, a certification question refers to the algorithmic task of finding an
 \emph{efficiently certifiable} upper bound on the given input. Via moment methods, the largest independent set have its size tightly concentrated at $\Theta(\frac{n\log d}{d}) $ on both distributions of \Erdos-\Renyi and $d$-regular. However, on the algorithmic side, the best known algorithm \cite{hoffman70} employ spectral techniques, and for both distributions yield a bound of $O(\frac{n}{\sqrt{d}})$ which is essentially off by an $O(\sqrt{d})$ factor from the ground-truth. It is a major open question whether the gap of $O(\sqrt{d})$ may be closed, and the independent set problem has been serving as a prototypical example of information-v.-computation gap in the statistics inference community. In this context, Sum-of-Squares semidefinite programming hierarchy arises as a testbed for algorithmic intractability for average-case problems as the general theory of $\NP$-hardness does not apply. Capturing spectral techniques that give the sated $O(\frac{n}{\sqrt{d}})$ bound,  SDP hierarchy is a natural candidate to turn to for a possibly improved upper bound.  The Sum-of-Squares hierarchy of SDP relaxations~\cite{Lasserre:2000:GOP:588888.589038,parrilo2000structured} is a hierarchy of increasingly powerful families parameterized by the \emph{SoS-degree}: it captures the best known algorithms for many worst-case combinatorial optimization problems \cite{GW94:stoc, AroraRV04} and yield improved and even optimal algorithms for statistical inference problems in the average-case setting \cite{BRS11, GS11, HSS15, MSS16, RRS17, KSS18, HL18, klivans2020efficient, hopkins2019mean, BK20}. As a result of the power of Sum-of-Squares in algorithmic design, establishing lower bounds against it for statistical inference problems has been a major research direction, and such hardness is usually seen as strong evidence for algorithmic intractability.

Starting from the seminal work of \cite{BHKKMP16} that resolves the planted clique problem within the Sum-of-Sqaures framework  in the last decade, there is a line of work that continues to investigate the limit of these algorithms on random graphs \cite{MRX20, PR20, Pang21, JPRTX, JPRX23, KPX24}: specifically, our work is directly inspired by  \cite{JPRTX, KPX24} that study the independent set problem on $G(n, \frac{d}{n})$,  and a concrete question we ask is whether one can switch Sum-of-Squares lower bounds from \Erdos-\Renyi to random regular graphs.

  The particular aforementioned results do not extend to the $d$-regular graph setting and it is cast as an explicit open question whether one may obtain similar results in $G_d(n)$. Despite several results for basic SDP on random $d$-regular graphs \cite{BKM19, deshpande2019threshold} and such results being anticipated given the results on \Erdos-\Renyi, not much is known for higher degree Sum-of-Squares in the regular graph setting (beyond degree-$4$). In fact,
 it was far from clear prior to this work how existing techniques from the i.i.d. setting may translate for the regular graph distribution, and what kind of new technical obstacles may arise. For example, it takes considerably technical efforts to improve the original result of \cite{BHKKMP16} to satisfy the clique-size constraint exactly \cite{Pang21}, and it is conceivable that similar challenges may appear for the regular graph setting.

 From this perspective, our work completes the missing piece in the paradigm of Sum-of-Squares hardness by extending its coverage to random regular graphs,  and develops tools to enable a smooth transfer within the known framework for applications in contexts beyond max independent set. In particular, we believe our norm bound result would also enable a smooth transfer for applications outside Sum-of-Squares lower bounds by considering spectral algorithms for average-case problems, and we leave this direction for future work.

\paragraph{Works for Low-Degree Polynomial Hardness and Level-$2$ Local Statistics}

Closely related to works in Sum-of-Squares lower bounds, low-degree polynomial framework is also another active platform for examining average-case hardness \cite{kunisky2019notes, wein2020optimal, SW22, gamarnik2022hardness, rush2022easier, Kothari2023IsPC} . Random regular graphs have been investigated in this model \cite{BBKMW20}, however, as far as we understand, a general theory for regular instances in the low-degree framework remains mostly at large due to the absence of an explicit orthogonal basis. That said, the recent work of \cite{kunisky2024computationalhardnessdetectinggraph} considers a close variant of our problem in the local statistics hierarchy of degree $LoSt(2,D)$, and asks explicitly for results of hierarchy of higher degree $LoSt$ in regular graphs. On a high level, the local-statistics hierarchy may be viewed as an adaptation of Sum-of-Squares hierarchy to distinguishing problems, and $LoSt(2,D)$  may be seen as a hybrid of hardness against degree-$D$ polynomials and the basic SDP (which is itself captured by degree-$2$ SoS).  Our result addresses their question in the context of independence number.

\paragraph{Previous Works for Graph Matrix Norm Bounds} The notion of graph matrix is first formally coined in \cite{BHKKMP19} where its norm bounds are given in \cite{BHKKMP19} and generalized in \cite{AMP20} for the dense setting for random variables with bounded Orlicz norm. In the sparse random graph setting, for random graphs of average degree at least $\polylog(n)$, rough norm bounds are studied in \cite{JPRTX} and \cite{RT22}. More recently, our previous work \cite{KPX24} optimizes the subpolynomial dependence of the norm bounds which are crucial for the sparse setting, and gives almost tight bounds in the setting of bounded average degree. Our work showcases the versatility of the framework in \cite{KPX24} by adapting it to the regular graph setting. We believe that following our analysis here, a careful 
  combination of some analog of the main lemma from \cite{dregspecgap} in the sparse regime (implicit in \cite{Bordenave2015ANP}) and \cite{KPX24} would yield norm bounds for the sparse regime for $d=O(1)$ (albeit in a slightly different model for random regular graphs), and we leave this for future work.

\paragraph{Distribution of Random $d$-Regular Graphs $G_d(n)$ } The spectrum and the spectral norm of random regular graph is an important topic in random matrix theory, and there have been a long line of works spanning over decades (eg. \cite{hoffman70, 4568282, Fri08, https://doi.org/10.1002/rsa.20406, BLM15, specgapdensegraph, Bor19, huang2021spectrum, chen2024newapproachstrongconvergence} and we refer reader to see references therein).
Before we dive into technical analysis for Sum-of-Squares, let us first make clear the specific distribution of random regular graphs $G_d(n)$ that we are working with throughout this work.    We start by considering the set of all graphs on $n$ vertices such that each vertex has degree $d$. Since such set is non-empty whenever $nd$ is even and $0<d<n$, we may then impose a uniform measure on this set, and define $G_d(n)$ the uniform distribution of graphs drawn from this set.

Our choice of distribution is slightly different from the usual choices of configuration/lifting model studied in the previous works for low-degree polynomial analysis. Interestingly, these two alternatives are usually used for the ease of moment calculation as it is possible to employ direct probabilistic calculations in these models.  However, we believe such distinction is rather technical and should not be conceptually significant:
  both distributions are both contiguous to the uniform distribution of random $d$-regular graphs. In fact, our specific choice is picked so that we may adopt the moment calculation from previous works in a blackbox manner, while we believe similar calculation on either model is true, and  can be extracted so that our results for norm bounds may be plugged in mechanically again. That said, to keep the presentation of our work concise, and to maintain the focus on our switching analysis, we opt to restrict our attention to the distribution $G_d(n)$. 

\subsection{Informal Statement of Our Results}

We now give a quick overview and an informal statement of results. We give matrix norm bounds for random $d$-regular graphs. The results are two-fold. 

On the one hand, we identify the "crucial" combinatorial structure of the underlying shape that dictates whether the input distribution of \Erdos-\Renyi matters, i.e., when the same norm bound holds on both distributions. On the other hand, we identify the tight (up to lower order dependence) dependence of graph matrix norm bound for shapes that have different spectral norm upper bounds in these two distributions.

\begin{theorem}[Informal of \cref{thm:normbound}]
	For shapes $\tau$ that do not have any edge "floating", i.e., edge-component not connected to the left/right matrix boundary , the same graph matrix norm bound (in \cite{JPRTX})   for \Erdos-\Renyi continues to hold with high probability for the random $d$-regular graph distribution $G_d(n)$.
	
	For shapes $\tau$ that have some floating component that is additionally tree-like,  the graph matrix norm bounds for \Erdos-\Renyi no longer hold in $G_d(n)$, and our norm bound for $G_d(n)$ gives a $\sqrt{n}$ factor blow-up over the prior norm bound is  for each such component .
\end{theorem}

\section{Graph Matrix Norm Bounds for Random Regular Graphs}

In this section, we describe our key technical ingredient, graph matrix norm bounds for random regular graphs. Even in the i.i.d. setting, the usual challenge of analyzing spectral norm bounds for graph matrix is the dependence of entries as a graph matrix of dimension $N \times N$ typically has entries given by polynomial functions of $N^{o(1)}$ bits of independent underlying randomness. In the regular graph setting, the $N^{o(1)}$ bits of underlying randomness are further correlated with each other. The further correlation poses challenge to the previous methods of analyzing graph matrix norm bounds via trace moment method, and our main contribution is to build upon previous known techniques for the i.i.d. setting, and to handle the correlation posed by regularity.

\subsection{Preliminaries of Graph Matrix}
Let us now start with the formal definition of graph matrix. Usually in the case of i.i.d. input, graph matrix is defined using the corresponding orthogonal basis for the null distribution, which is the $p$-biased basis for the Erdos-Renyi distribution $G(n,\frac{d}{n})$. The specific basis is defined using the following Fourier characters,

\begin{definition}[Fourier character for $G(n,\frac{d}{n})$] Let $\chi$ denote the $p$-biased Fourier character, we have \[ 
\chi(1) = \sqrt{\frac{1-p}{p}} \approx \sqrt{\frac{n}{d}} \,,
\]
and \[ 
\chi(0) = -\sqrt{\frac{p}{1-p}} \approx -\sqrt{\frac{d}{n}}\,.
\]
For a subset of edges $S \subseteq \binom{n}{2}$, we denote its corresponding basis element as $\chi_S(G) = \prod_{e\in S} \chi_e(G)$. Throughout this work, we use $p$ and $\frac{d}{n}$ interchangeably unless otherwise specified.
\end{definition}
To see that this forms the an orthogonal basis for \Erdos-\Renyi $G(n,\frac{d}{n})$, it is straightforward to verify that each edge is independent of each other, and each edge character has mean $0$ and variance $1$. In this work, even though we are working with input coming random $d$-regular graph, we define the underlying edge-character as the above Fourier character for \Erdos-\Renyi graph of $G(n,\frac{d}{n})$. Crucially, we want to point out that that the edge characters are no longer, even though close to being, orthogonal. From this point on, we will use $G$ to refer to the graph input in general, and as we use the same edge character, we shall not make the distinction of the distribution from which $G$ is sampled.

Next, we are ready to introduce \emph{shape}, a representation of structured matrices whose entires are polynomial functions of the underlying graph input $G$. 

\begin{definition}[Shape]
\label{def:shape}
    A shape $\tau$ is a graph on vertices $V(\tau)$ with (possibly multi-)edges $E(\tau)$ and 
    two ordered tuples of vertices $U_{\tau}$ (left boundary) and $V_{\tau}$ (right-boundary). 
    We denote a shape $\tau$ by $(V(\tau), E(\tau), U_{\tau}, V_{\tau})$. 
    
\end{definition}
We remind the reader that $V_\tau$ should be distinguished from $V(\tau)$, where $V_\tau $ is the right boundary set, while $V(\tau)$ is the set of all vertices in the graph.

\begin{definition}[Shape transpose]
    For each shape $\tau$, we use $\tau^T$ to denote the shape obtained by flipping the boundary $U_\tau$ and $V_\tau $ labels. In other words, $\tau^T = (V(\tau), E(\tau), U_{\tau^T} = V_\tau, V_{\tau^T} = U_\tau)$.
\end{definition}

 \begin{definition}[Embedding]
Given an underlying random graph sample $G$,  a shape $\alpha$ and an injective function $\psi(\al): V(\alpha) \to [n],$ we define $M_{\psi(\al)}$ to be the matrix of size
$\frac{n!}{(n - |U_{\alpha}|)!} \times \frac{n!}{(n - |V_{\alpha}|)!}$
with rows and columns indexed by ordered tuples of $[n]$ of size $|U_{\alpha}|$ and $|V_{\alpha}|$ with a single nonzero entry \[
      M_{\psi(\alpha)}[\psi(U_\al), \psi(V_\al)] = \prod_{(i,j)\in E(\al)}  \chi_G\left(\psi(i), \psi(j)\right).
     \]
     and $0$ everywhere else.
     \end{definition}

\begin{definition}[Graph matrix of a shape]
    For a shape $\alpha$, the graph matrix $M_\alpha$ is
    \[M_\alpha = \displaystyle\sum_{\text{injective }\psi: V(\alpha) \to [n]} M_{\psi(\alpha)}. \]
     When analyzing the Sum of Squares hierarchy, we extend $M_{\alpha}$ to have rows and columns indexed by all tuples of vertices of size at most $\frac{\dsos}{2}$ by filling in the remaining entries with $0$.
\end{definition}

In words, for a graph matrix of shape $\tau$, its entry of $M_\tau[S,T]$ for set $S, T$ agreeing with the boundary condition of $\tau$ is given by the summation over the injective labeling of middle vertices in $V(\tau)\setminus U_\tau\cup V_\tau$.

 \begin{figure}[ht!]
    \centering
    \begin{subfigure}[b]{0.25\textwidth}
        \includegraphics[width=\textwidth]{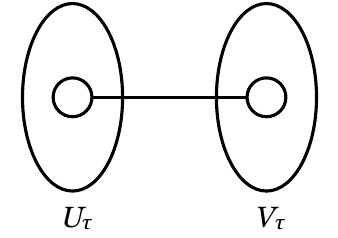}
        \caption{Line Graph. }
        \label{fig:GOE}
    \end{subfigure}
    \quad
    \begin{subfigure}[b]{0.25\textwidth}
        \includegraphics[width=\textwidth]{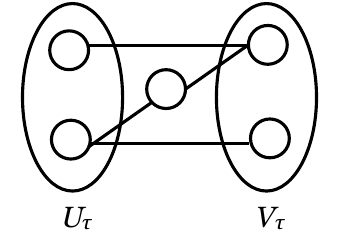}
        \caption{Z-shape.}
        \label{fig:Z-shape}
    \end{subfigure}
    \quad
\end{figure}
Take the following two diagrams for example. 
 The first diagram is the usual adjacency matrix in the $p$-biased basis (with diagonal removed), and has entries $M_{\text{line}}[i,j] = \chi_G(\{i,j\})$ where we remind the reader for $\chi_G(\{i,j\})$ is the $p=\frac{d}{n}$-biased character. The second diagram is a more interesting example of graph matrix with non-trivial correlations across the entries: it corresponds to a matrix of dimension $n(n-1)\times n(n-1)$ and has entries $M_Z[(a,b), (c,d)] = \sum_{t\notin \{a,b,c,d\}} \chi_G(\{a,b\}) \cdot \chi_G(\{b,t\})\cdot \chi_G(\{c,t\}) \cdot \chi_G(\{c,d\})$. Both matrices have $0$ at the undefined entries. 
 
 Before we state our main result for graph matrix norm bounds, we make one further definition for graph matrix that will crucially reflect the difference of norm bounds for the i.i.d. distribution and for the regular graphs. 
 \begin{definition}[Isolated vertex]
 	For a shape $\tau$ and a vertex $v\in V(\tau)\setminus U_\tau\cup V_\tau$, we call it \emph{isolated} if it is not incident to any edge. However, vertices in $U_\tau\cup V_\tau$ are never considered isolated.
 \end{definition}
 
 \begin{definition}[Floating component]
 	For a shape $\tau$, and an connected component $C$ induced by edges, we call $C$ \emph{floating} if there is no path from $C$ to $U_\tau\cup V_\al$.
 \end{definition}
 
 See the following for example of floating component. This is an $n\times n$ graph matrix of entries $M[i,j] = \chi_{G}(\{i,j\}) \cdot (\sum_{\{a,b\} \cap \{i,j\}= \emptyset, a\neq b} \chi_{G}(\{a,b\}) $. We call the edge $\{a,b\}$ floating in this shape as it is not connected to neither left nor right boundary.
 \begin{figure}[ht!]
    \centering
    \begin{subfigure}[b]{0.25\textwidth}
        \includegraphics[width=\textwidth]{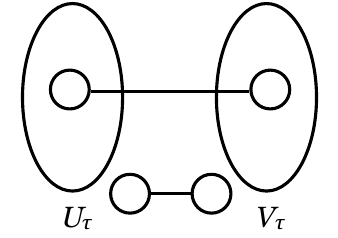}
    \end{subfigure}
    \caption{Floating Component }
\end{figure}
 
 	Importantly, observe that an isolated vertex of degree $0$ in $V(\al)\setminus U_\al\cup V_\al$ is not floating as we consider edge-induced component in the above definition. Finally, we make clear of the definition of a component being tree-like even though it may already be clear as its name suggests,
 \begin{definition}[Tree-like component]
 	We call a connected component $C $ with vertices $V(C)$ and edges $E(C)$ tree-like if $  V(C)  = E(C)+1$. In other words, there is no cycle in the connected component.
 \end{definition}	
 
 \subsection{Main Theorems of Graph Matrix Norm Bounds on Random $d$-Regular Graphs} 
 For some $\delta>0$ and for any shape $\tau$ with $|V(\al)| \leq n^\delta$, and $G_d(n)$ the distribution of random $d$-regular graphs for any $\log^{10} n \ll d = o(n)$ and $p=\frac{d}{n}$
  , we obtain the following theorems,
 \begin{theorem}[User-friendly version for shapes without floating component] For a shape $\tau$ without floating component, we define \[ 
 B_q(\tau ) \coloneqq \max_{S:\text{separator for } \tau} (\sqrt{n} \cdot q)^{|V(\tau)\setminus S|} \cdot \left(\sqrt{\frac{1-p}{p}}  \right)^{|E(S)|} \cdot \sqrt{n}^{|I(\tau)|} \cdot (\cnorm)^{|E(\tau)|} 
 \]
 for some constant $\cnorm, c>0$ and $I(\tau)$ the set of isolated vertices in $\tau$. 
 For any $\eps>0$, there exists some large enough constant $c>0$ such that for any $q=\Omega(|U_\tau||\log^c n)$ and $q<\min(d^{1/10}, n^{O(\delta)})$, \[ 
 \E_{G\sim G_d(n)} \left[\Tr((M_\tau M_\tau^\top)^q)\right] \leq  \left( (1+\eps)\cdot B_{q}(\tau)\right)^{2q}
 \]	
 \end{theorem}

\begin{theorem}[Full version for shapes containing floating component] \label{thm:normbound}  For a shape $\tau$ possibly containing floating components, we define \[ 
\mathsf{float}(\tau\setminus S) =\prod_{C_i:\text{ floating component not connected to } S  } \sqrt{n}\cdot 1[C_i \text{ is tree-like}] \,,
\]
for any $S\subseteq V(\tau)$ where $1[C_i \text{ is tree-like}] $ is the indicator for whether edge-induced component $C_i$ is tree-like, 
and define
\begin{align*}
	& B_q(\tau )\coloneqq \max_{S:\text{separator for } \tau} (\sqrt{n} \cdot q)^{|V(\tau)\setminus S|} \cdot \left(\sqrt{\frac{1-p}{p}}  \right)^{|E(S)|} \cdot \sqrt{n}^{|I(\tau)|} \cdot \mathsf{float}(\tau\setminus S) \cdot (\cnorm)^{|E(\tau)|  } 
	\end{align*} 
 for some constant $\cnorm, c>0$ and $I(\tau)$ the set of isolated vertices in $\tau$. 
 For any $\eps>0$ and for any $q=\Omega(|U_\tau|\log^c n)$ and $q<\min(d^{1/10}, n^{O(\delta)})$, \[ 
 \E_{G\sim G_d(n)} \left[\Tr((M_\tau M_\tau^\top)^q)\right] \leq \left((1+\eps)\cdot B_{q}(\tau)\right)^{2q}
 \]	
 \end{theorem}
 
 Let us now unpack the above norm bounds for clarity. Towards the end, we recall the norm bound from \cite{JPRTX} for the i.i.d. setting as\[
  \tilde{O}(\max_{S:\text{separator}} \sqrt{n}^{|V(\tau)\setminus S|} \cdot \sqrt{n}^{|I(\tau)|} \cdot \sqrt{\frac{1-p}{p}}^{|E(S)| } )\] 
  The main difference between the norm bound for shape with and without floating component is the extra factor captured by $\mathsf{float}(\tau\setminus S)$, that is each tree-like floating component disconnected from the separator $S$ requires an extra factor of $\sqrt{n}$ in the candidate upper bound. Despite the difference incurred by floating component, the candidate bound stated above are essentially identical as that for the i.i.d. setting in \cite{JPRTX} up to subpolynomial dependence.
 
As a corollary, we obtain the following matrix norm bound that holds with probability at least $1-o_n(1)$, $
\|M_\tau\| \leq (1+o_n(1)) \cdot B_q(\tau)
$.

\subsection{Trace Moment Method and Block-Value Bound}

\paragraph{Spectral Norm Bound from Trace Moment Method}  Our proof proceeds by analyzing trace moments of matrices using that for any $A$, the spectral norm $\|A\|_{spec} \leq \tr((AA^{\top})^q)^{1/2q}$. Taking $q$ to be logarithmic in the dimension of $A$ suffices to get a bound on $\|A\|$ sharp up to $1+\delta$ for an arbitrarily small $\delta>0$. Throughout this work, unless otherwise specified, for a matrix $A$, we denote its spectral norm by $\|A\| \coloneqq  \|A\|_{\text{spec}}$. Expanding the trace power for a graph matrix $M_\tau$ gives
\begin{align*}
	\tr((M_\tau M_\tau^{\top})^q) &= \sum_{\substack{P:\text{a shape-walk of length }q \\P =\{S_1, T_1, S_2, \dots, S_{q-1}, T_{q-1}\}\\ S_i:\text{labeling of left boundaries of $\tau$} \\ T_i:\text{labeling of right boundaries of $\tau$}} } M_\tau[S_1,T_1]  M_\tau^\top[T_1, S_2] \dots M^\top[T_{q-1}, S_q=S_1]\,.
\end{align*}

We can now adopt a combinatorial perspective. The trace power is expanded as a sum of weighted walks of length $2q$ over the entries of $M_{\tau}$, and each walk can be viewed as graph consisting of $q$ blocks of $\tau$ and $\tau^T$. Each term in the expansion of $\tr((M_{\tau}M_{\tau}^{\top})^q)$ corresponds to a labeling of the vertices in $q$ copies of $\tau$ and $\tau^{\top}$ with labels from $[n]$ (i.e, vertices of the graph $G$) satisfying some additional constraints that correspond to a valid walk. The following definition captures these constraints. 

\begin{definition}[Shape walk and its valid labelings]
	Let $\tau$ be a shape and $q\in \N$. For each walk $P$, let $G_P(\al,2q)$ be the \emph{shape-walk} graph of the shape-walk $P$ vertices on vertices  $V_P(\tau,2q)$ formed by the following process, \begin{enumerate}
		\item Take $q$ copies of $\tau_1,\dots,\tau_q$ of $\tau$, and $q$ copies of $\tau_{1}^\top,\dots, \tau_q^\top$ of $\tau^\top$;
		\item For each copy of $\tau_i$ (and $\tau_i^T$), we associate it with a labeling, i.e., an injective map $\psi_i : V(\tau) \rightarrow [n]$ (and respectively $ \psi_i' $ for $\tau^\top$); 
		\item For each $i\in [q]$, we require the boundary labels to be consistent as ordered tuples, i.e. $\psi_i(U_{\tau}) =\psi'_{i-1}(V_\tau)$ and $\psi'_{i}(U_{\tau^\top}) = \psi_i(V_{\tau_i})$;
		\item Additionally, each such walk must be closed, i.e., $\psi'_q(V_{\tau^\top}) = \psi_1(U_\tau)$ as a tuple-equality;
		\item We call each $\tau_i$ and $\tau^\top_i$ block a \emph{BlockStep} in the walk.
\end{enumerate}
For each walk-graph , we associate it with a natural decomposition $P = \psi_1\circ \psi_1'\circ \psi_2\circ\dots\circ\psi_{q}\circ \psi_q'$.
\end{definition}

Moreover, we let $E(P)$ be the set of edges used in the labeled walk $P$.
Before we proceed further, we would like remind the reader not to confuse vertices/edges in the walk with vertices/edges in the shape. The vertices in a walk are “labeled” by elements in $[n]$. Similarly, each edge $e\in E(P)$ in a walk is labeled by an element in $\binom{n}{2}$ . We will use the terms “labeled vertex” and “labeled edge” unless it is clear from context.

Throughout our work, we require spectral norm bounds on random matrix that holds with probability at least $1-o(1)$, and this is further accomplished by bounding the expectation of the above trace quantity and a simple application of Markov's inequality. Taking the expectation over the randomness of the underlying input $G$, and reorganizing the terms on the righthand side by grouping steps by their underlying labeled-edge,
we can write the expected trace as \begin{align*}
	\E_G \tr((M_\tau M_\tau^{\top})^q) &= \E_G\sum_{\substack{S_1, T_1, S_2, \dots, S_{q-1}, T_{q-1}\\ S_i:\text{labeling of left boundaries of $\tau$} \\ T_i:\text{labeling of right boundaries of $\tau$}} } M_\tau[S_1,T_1]  M_\tau^\top[T_1, S_2] \dots M_\tau^\top[T_{q-1}, S_1]\\
	&= \E_G \sum_{\substack{P \\ \text{shape-walk of length } q}} \prod_{e\in E(P)} \chi_G(e)^{\mul_P(e)}
\end{align*} 
where we use $\mul_P(e)$ to denote the multiplicity of labeled-edge $e$ in the walk $P$.
In the i.i.d. setting, the above quantity further factorizes over edges as \[ \E_G \sum_{\substack{P\\\text{shape-walk of length } q}} \prod_{e\in E(P)} \chi_G(e)^{\mul_P(e)} =\sum_{\substack{P\\ \text{shape-walk of length} q}}  \prod_{e\in E(P)} \E_G[\chi_G(e)^{\mul_P(e)}]\,,
\]
and a crucial observation at this point 
is that for each walk to give non-zero contribution, each labeled edge needs to appear in at least two steps, as otherwise the expectation of the walk factorizes over the edges and edges appearing a single time gives mean $0$. 

Noting that such immediate factorization is not true in our setting of random regular graphs, however, we will proceed with the i.i.d. setting to describe the machinery of obtaining matrix norm bound from block-value bound from \cite{AMP20, JPRTX, HKPX23, KPX24} , and then outline challenges and solution posed by regularity in the subsequent section. 

In particular, our norm bound techniques may be viewed as a showcase for the versatility of the machinery developed in \cite{KPX24} with simplification as we do not optimize over polylog dependence which is the main target of their work, and with adaptation to random regular graphs.

\paragraph{Spectral Norm Bound from Block Value Bound}
Noting that the expected trace is a weighted count of walks, trace moment methods for random matrix, in particular adjacency matrix and its variants, usually proceed in a global fashion by analyzing the global constraints of the walks (eg. each edge needs to appear twice) via an encoding argument and then carefully analyzing the combinatorial factors therein. However, as noted in previous works, such analysis may easily become overwhelming if applied for non-trivial graph matrices. 

To handle the potential complication from a global analysis, the machinery of block value bound proposes to find a \emph{local} bound that holds for each step of the walk, or more concretely when applied for graph matrix, for each \emph{BlockStep} of a \emph{shape walk}. For starters, one may first unpack the expected trace, and observe that there are two kind of factors contributing to the final expected trace,\begin{enumerate}
	\item \textbf{Vertex Factor}: the combinatorial "counting" factor that specifies each step of the walk;
	\item \textbf{Edge Value}: the analytical factor from edges (the underlying randomness) which corresponds to the term  $\E \prod_{e\in E(P)} \chi_G(e)^{\mul_P(e)}
$.
\end{enumerate} 

To be concrete, one may think of the vertex factor as the factor needed for the decoder to identify the walk in the usual encoder-decoder setting. For example, a typical vertex requires a cost of $[n]$ if it is the first time it appears in the walk as there are $n$ possible choices for the vertex label. Within this scheme, the vertex-factors altogether should be sufficient for the decoder to decode the walk, and when combined with the analytic factor from the underlying randomness, gives an upper bound for the desired expected trace.

With these two types of factors in hand, the core message of the machinery is that one can split the factors in an essentially "even" manner so that the global bound reduces to finding a local bound of the BlockStep. In words, for each BlocStep, the machinery demands a factor assignment scheme such that for the candidate upper bound $B(\tau)$, one can bound the contribution of a given $\text{BlockStep}_i$ to the global trace by  \[ 
\mathsf{vtxcost}(\text{BlockStep}_i)  \cdot \edgeval(\text{BlockStep}_i) \leq B(\tau)
\] 
on a high level.

 The target of the block value machinery is then formally defined as \begin{definition}(Block value bound) Fix $q\in \N$ and a shape $\tau$. For any vertex/edge factor assignment scheme, we call $B_q(\tau)$ a valid block-value function for $\tau$ of the given factor-assignment scheme if \[ 
\E\left[\tr(M_\tau M_\tau^\top)^q\right] \leq f(\tau) \cdot B_q(\tau)^{2q}
\]
for some auxiliary function $f(\tau) $ s.t. $f(\tau)^{1/2q} = 1+o_n(1)$, 
and for each $\mathsf{BlockStep}_i$ throughout the shape walk, the following holds, \[ 
\mathsf{vtxcost}(\mathsf{BlockStep}_i) \cdot \mathsf{edgeval}(\mathsf{BlockStep}_i) \leq B_q(\tau)
\]	
where $\text{vtxcost}(\mathsf{BlockStep}_i), \mathsf{edgeval}(\mathsf{BlockStep}_i)$ are the corresponding vertex-cost and edge-value assigned to $\mathsf{BlockStep}_i$ by the given factor-assignment scheme.
\end{definition}

\begin{remark}
	To break the assymetry between the first block-step and any subsequent block-step, we take $f(\tau) =O( n^{|U_\tau|})$ to capture of identifying the start of the walk, and then walk-value may then be bounded by treating the first block as any subsequent block. 
\end{remark}  

The intended upper bound $B_q(\tau)$ in the application may have extra dependence on $q$ on the length of the walk, while this effect only culminates to subpolynomial $\tilde{O}(1)$ factors of our norm and can be ignored for the purpose of our work as we work in the regime of $d>\log^c n$ for some absolute constant $c>0$.  We ignore the subscript $q$ in the block-value bound from this point on.

To further enable a local analysis and break correlations across the BlockSteps, we design our factor assignment (for both vertex-cost and edge-value) that only depends on the step-labels of the block.  We formally define step-label as the following,
\begin{definition}[Step-label and BlockStep-Label]
  We label each step (i.e. edge's multiplicity) as the following,
  \begin{enumerate}
  	\item \textbf{Singleton Step}: For any edge that appears only once throughout the walk, we call this edge/step a singleton step/edge;
  	\item \textbf{F/S Step}: For steps along a non-singleton edge (i.e. an edge that appears more than once), and the edge is appearing for the first time in the walk, we call it an $F$ (fresh) step if it leads to a new, unexplored vertex in $[n]$; 
  	\item \textbf{H Step}: For a step along  a non-singleton edge that appears for neither the first nor the last time in the walk, we call it an $H$ (high-mul) step;
  	\item \textbf{R Step}: For a step along  a non-singleton edge that appears for the last time in the walk, we call it an $R$ (Return) step.
  	  \end{enumerate}
  	  A step-labeling for a block-step assigns step-labels for each edge in the block, and can be represented as a function $L:E(\tau) \rightarrow \{F, S, R, H, \text{Singleton} \}  $.
  \end{definition}

The following proposition shows that we can easily obtain a valid $B_q(\tau)$ once we have an appropriate factor assignment scheme.

\begin{proposition} \label{prop:sum-of-labelings}
    For any graph matrix $M_\tau$ and any valid factor assignment scheme, for any step-labeling $L$ of $\tau$, define \[ 
    B(L) \coloneqq \mathsf{vtxcost}(L) \cdot  \mathsf{edgeval}(L)\,,
    \]
    then
    \begin{align*}
		B_q(\tau) \coloneqq \sum_{L : \text{ step-labelings for } E(\tau)} B(L) =  \sum_{L : \text{ step-labelings for } E(\tau)} \mathsf{vtxcost}(L) \cdot \mathsf{edgeval}(L)
	\end{align*}
    is a valid block-value function for $\tau$.
\end{proposition}
\begin{proof}
It is immediate to observe that $B_q(\tau)$ defined above is a local bound for each block-step. 
To see the global upper bound, observe that the trace can be bounded by the matrix dimension (specifying the start of the walk captured by the auxiliary function $f(\tau)$) times
    \begin{equation*}
        \sum_{\substack{L_1,\dots, L_{2q}: \\ \text{step-labelings for $E(\tau)$}}} \prod_{i=1}^{2q} \mathsf{vtxcost}(L_i) \cdot \edgeval(L_i)
        \leq \left( \sum_{L:\text{step-labelings for } E(\tau)} \mathsf{vtxcost}(L) \cdot \mathsf{edgeval}(L) \right)^{2q} 
        \qedhere
    \end{equation*}
\end{proof}

The next step of the Block-Value bound machinery is then to identify a factor-assignment scheme (for both vertex-cost and edge-value) that produces a minimal upper bound of $B(\tau)$. It should be pointed out the most technical component of the machinery is usually to design the assignment scheme for vertex-costs
and carefully analyze the interplay between vertex-cost bound and edge-value bound,
 while the edge-value scheme is rather immediate in the i.i.d. setting.
 
  However, in our setting of regular graphs, the edge-value component is a major focus of our work as this is where the distinction of the underlying graph distribution kicks in.  
\subsection{Bounding Walk Value: putting in the regularity}
Before we consider how the edge-value may be factorized into each BlockStep, a crucial starting point is to understand what the edge-value bound globally for a whole walk. In particular, we want to find an upper bound of edge-value for $
\E_{G\sim G_d(n)} \underbrace{\prod_{e\in P} \chi_G(e)^{\mul_P(e)}}_{\coloneqq \mathsf{edgeval}_{\text{scaled}}(P)} 
 $ for a "not-too-large" subset of edges $P\subseteq \binom{n}{2}$ with various multiplicities in a random regular graph. 

Fortunately, such question is well-understood in the context of random regular graphs, and 
we use result from the previous work of \cite{dregspecgap} in the context of the spectral gap for random $d$-regular graph. Before we state the known result, we first make the following additional definitions,
 
\begin{definition}[Surprise step/visit]
	In a shape walk, we call a step a $\emph{surprise visit} $ if it leads to a vertex that has already been visited/explored in the walk.
\end{definition}

\begin{definition}[Singleton step/edge]
	In a shape walk $P$, we call a step along labeled edge $e$ (of the underlying randomness) a singleton-edge/step if $e$ appears only once throughout the whole walk.\\
	On the other hand, we call a step using along a labeled edge $e$ that appears more than once a non-singleton step and the underlying edge a non-singleton edge.
\end{definition}
 We are now ready to recall the previous known bound, which bounds the unscaled value of a walk $P$ of length-$2q$ in a random $d$-regular as the following.
\begin{lemma}[Restated edge-value bound from Corollary 3.6 in \cite{dregspecgap}] Let $\log n \ll q \ll d^{1/10}$, and $d<cn$ for some constant $c>0$, define \[ \mathsf{edgeval}_{\text{unscaled}}(P) \coloneqq \prod_{e\in P} \underbrace{(G(e)-\frac{d}{n})^{\mul_P(e) }}_{\text{unscaled character}} \,,   \]
	where $G(e)$ is the $0/1$ indicator variable for edge in graph sample $G$,
	the following bound holds \[ 
	|\E_{G\sim G_d(n)} \mathsf{edgeval}_{\text{unscaled}}(P) | \leq  (2+o(1)) \cdot  (2q)^{2\cdot |\text{surprise(P)}|}\cdot  (p(1-p)/n)^{|\text{singleton(P)}|/2} (p(1-p))^{|\text{NonSingleton}(P)| }
	\]
	where $\text{surprise}(P)$ is the set of surprise visits in the walk $P$, and $\text{Singleton}(P)$ is the set of singleton edges (i.e., edges that appear only once in the walk), $\text{NonSingleton}(P)$ is the set of edges that appear at least twice, and importantly, $|\text{NonSingleton}(P)|$ denotes the number of such edges (not counting multiplicities). 
\end{lemma}

At this point, we shall emphasize that the most important aspect of the above bound is the factor of $(p(1-p)/n)^{|\text{singleton(P)}|/2}$ which accounts for the change that we are working in a non i.i.d. setting, and each edge appearing only once no longer has zero expected value as in the i.i.d. setting. This lemma does not follow from an immediate magnitude bound, and requires a careful analysis using cancellation of the terms that appear in the moment expansion which is arguably the main technical component of \cite{dregspecgap}: compared with magnitude bound, it comes with an extra $\frac{1}{\sqrt{n}}$ decay for each singleton edge without which our spectral norm bounds would degrade into a trivial Frobenius norm bound for various shapes.

 Before we combine this edge factor with combinatorial counting, let us switch back to the Fourier basis in which edge has their character given, and note that the scaling is picked such that  \[
 \edgeval_{\text{scaled}}(P) = \edgeval_{\text{unscaled}}(P) \cdot \sqrt{\frac{1}{p(1-p)}}^{\mul(P)} 
  \]
  where $\mul(P) \coloneqq \sum_{e\in P}\mul_P(e) $ is the total edge multiplicity of walk $P$, and $\mul_P(e)$ is the multiplicity of edge $e$ in $P$. Combining with the magnitude bound for the unscaled value, we have \begin{align*}
 &\left| \E_{G\sim G_d(n)} \edgeval_{\text{scaled}}(P)\right| \\&\leq   (2+o(1)) \cdot  (2q)^{2\cdot |\text{surprise(P)}|}\cdot  (p(1-p)/n)^{|\text{singleton(P)}|/2} \cdot (p(1-p))^{|\text{NonSingleton}(P)| } \cdot  \sqrt{\frac{1}{p(1-p)}}^{\mul(P)} 	
  \end{align*}

  With the above bound for the walk value, we arrive at our edge-value assignment scheme that assigns factors from the above upper bound to each edge's multiplicity in the walk. We follow the notation in previous works and label each edge's multiplicity (i.e. step) as the following,

  	At this point, it is crucial for the reader to distinguish a \emph{Singleton} step and an $S$ step of surprise visit. In fact, we advise for the readers to skip upon $S$ step of surprise visit as they only concern for subpolynomial dependence of the final norm bound which we do not optimize in this work.
  
  \begin{proposition}
  	The following step-value assignment scheme gives an upper bound for the edge-value up to a factor of $O(1)$:\begin{enumerate}
  		\item For singleton step/edge, assign a value of $\sqrt{\frac{p(1-p)}{n} }  \cdot \sqrt{\frac{1}{p(1-p)}} = \sqrt{\frac{1}{n}} $;
  		\item For an $F/S/R$ step of an edge that appears more than once, assign a value of $ \sqrt{p(1-p)} \cdot \sqrt{\frac{1}{p(1-p)}} = 1 $;
  		\item For an $S$ step,  i.e., the underlying edge is non-singleton and appearing for the first time wile the step leads to a visited vertex in the walk, we additionally assign a value of $(2q)^2$. That said, an $S$ step of a surprise visit gets assigned a value of $(2q)^2$ in total;
  		\item For an $H$ step of an edge that appears more than twice, assign a value of $
  	  \sqrt{\frac{1-p}{p}} 	  $. 
  	\end{enumerate}
  	Let $f(s)$ be the assigned value of step $s$ in according to the above scheme, \[ 
  	(2+o(1)) \prod_{s\in P: \text{step} } |f(s)| \geq | \edgeval_{\text{scaled}}(P) |
  	\] 
  	In words, up to the factor of $(2+o(1))$, the above step-value assignment scheme assigns the upper bound factor of the walk to each step.
  \end{proposition}
  \begin{proof}
  	We first observe that the factor $\sqrt{\frac{1}{p(1-p)}} $ gets assigned to each step regardless of label, which is consistent to the fact that we have $\sqrt{\frac{1}{p(1-p)}}^{\mul(P)} $ in the original upper bound. Additionally, for each edge that appears more than once, we pick up a factor of $p(1-p)$ and this is distributed among the $F//R$ steps of the edge. For edges that appear only once, the whole factor is assigned to the singleton edge. For the factor of $(2q)^{|surprsie(P|) }$, we assign them to the $S$-step in which the edge is making the first appearance.  \end{proof}

\subsection{Bounding Vertex-Factor: as if edges are i.i.d.}
For people familiar with the trace-moment method calculation, this is usually the most technical and cumbersome step of most random matrix norm bounds. For this component of our work, we largely draw from prior works \cite{JPRTX, HKPX23, KPX24} for applying the module of assigning vertex-factors while the effect of having singleton edges comes into play later when combined with the edge-value assignment scheme.
\paragraph{Vertex-factor Assignment Scheme}
In the case of each edge appearing at least twice for the trace to be non-vanishing, we consider the following vertex-factor assignment scheme. 

	\begin{enumerate}
		\item We assign a factor of $n$ for a vertex making both of its first and last appearance at the same block; 
		\item otherwise, we assign a factor of $\sqrt{n \cdot q_\tau} \leq \sqrt{n} \cdot q_\tau$ for the first and last appearance (where we use a loose bound here for simplicity as we do not optimize the dependence on $q$);
		\item we assign a factor of $q_\tau$ for any of its middle appearance.
			\end{enumerate}

\begin{proposition}[Vertex-factor assignment scheme]
The above is a sound factor assignment scheme. In other words, 
let $\mathsf{vtxcost}(P)$ be an upper bound for the vertex-factor of all vertices throughout the walk $P$, 
\[ 
 \mathsf{vtxcost}(P) \leq \prod_{\substack{v: \text{vertex }v\\ \text{appearance }i  }  }| \text{vtxcost}(v,i) | 
\]
where $ \text{vtxcost}(v,i)$ is the factor assigned by the vertex-factor assignment scheme for vertex $v$ at its $i$-th appearance.
 	\end{proposition}
 \begin{proof}
 	In the case where there is no singleton edge, each vertex needs to appear at least twice. Further, for each vertex that appears more than once in a length-$q_\tau \coloneqq 2q|V(\tau)|$ walk, its first appearance requires a cost of $[n]$ and any of its subsequent appearance can be specified at a cost of $q_\tau$.  The crux of this scheme is to split the factor of $n\cdot q_\tau $ evenly to its first and last appearance, and this proves the factor assignment schemes is a sound upper bound for vertices that do no make first and last appearance in the same block (i.e., a single appearance throughout the walk).
 	
 	In the case of singleton edges, we note that for vertices arrived by a singleton edge on the $V_\tau$ boundary, such vertex continues to make appearance in two blocks and that is still a factor of $\sqrt{n\cdot q_\tau} \leq \sqrt{n} \cdot q_\tau$ per appearance. For vertices arrived by a singleton edge but in the interior, i.e., $V(\tau)\setminus V_\tau$, we follow the notation that a vertex is allowed to make more a single first/last appearance in the same block. That said, such a vertex may be making its first and last appearance in the same block, and gets assigned a factor of $n\cdot q_\tau$ by the previous reasoning. However, notice that for such vertices, a cost of $[n]$ at the first appearance suffices, therefore we can spare a factor of $q_\tau$ that is originally intended for its subsequent appearance, and assign a cost of $n$ in total.   
 	 \end{proof} 
 
  \subsection{Explicit Application of BlockValue Machinery}
  In this section, we apply the above machinery to the explicit graph matrices of line graph (see following for diagrams) to familiarize readers with our described scheme, and show one can easily recover the $\tilde{O}(\sqrt{n})$ bound for centered adjacency matrix, which is equivalent to a $\tilde{O}(\sqrt{d})$ bound for the spectral gap of a random graph.
  \begin{figure}[ht!]
    \centering
    \begin{subfigure}[b]{0.25\textwidth}
        \includegraphics[width=\textwidth]{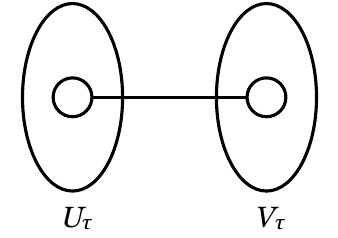}
        \label{fig:GOE}
    \end{subfigure}
    \end{figure}
\begin{lemma}
For the line-graph $\tau_{line}$, we have \[ 
B(\tau_{line}) \leq \tilde{O}(\sqrt{n})
\,,
\] 
as a corollary, we have $\|M_{line}\| \leq O(\sqrt{n}\poly\log(n))$ w.h.p. 	
\end{lemma}
\begin{proof} We first take $q=\polylog n$.
	To apply block-value bound, fix a traversal direction from $U_\tau$ to $V_\tau$, and given current boundary at $U_\tau$, we start by casing the step-label. If this is an $F$-step \begin{enumerate}
	\item $V_\tau $ is making a first appearance and a factor of $\sqrt{n} \cdot q$ is assigned;
		\item  $U_\tau$  is making a middle appearance and at most a cost of $q$ is assigned to the vertex;
				\item The edge is non-singleton and making the first appearance, which gets value $1$;
		\item This combines to a value of $\sqrt{n}\cdot q$ in total.
	\end{enumerate} 
	Similarly if this is an $R$-step, \begin{enumerate}
	\item $V_\tau $ is potentially making a last appearance and at most a factor of $\sqrt{n} \cdot q$ is assigned (a factor of $q$ for middle appearance);
		\item  $U_\tau$  is making a middle appearance and at most a cost of $q$ is assigned to the vertex;
				\item The edge is non-singleton and making the first appearance, which gets value $1$;
		\item This combines to a value of $\sqrt{n}\cdot q$ in total.
	\end{enumerate} 
	If this is an $S/H$-step, we have at most a vertex factor of $q^2$ in total, and an edge-value of $1$ or in the case of $H$-step, a value of $\sqrt{\frac{1-p}{p}} \approx \sqrt{\frac{n}{d}}$. This gives a total factor of $q^2 $ and  $ \sqrt{\frac{1-p}{p}}\cdot q^2$.
	
	Lastly, if this is a singleton-step, we have vertex factor of at most $n$ in total when the left vertex is making a last appearance and the right vertex a first appearance. Crucially, the singleton-step gives a decay of $\frac{1}{\sqrt{n}}$, and therefore we have a combined bound of $\sqrt{n}$.
		
	Summing over the above step-labels, we have \[ 
	B(\tau)\leq 2\cdot \sqrt{n}\cdot q^2 +  q^2 + \sqrt{\frac{1-p}{p}} \cdot q^2 + \sqrt{n}
	\leq \tilde{O}(\sqrt{n})
	\]
as we take $q = \poly\log n$.
	\end{proof}

  \subsection{Connecting back to Vertex Separator}
The crucial quantity governing graph matrix norm bounds is the notion of vertex separator, defined as the following,
\begin{definition}[Vertex separator]
	For a shape $\tau$, we say $S \subseteq V(\tau)$ is a vertex separator for $\tau$ if any path from $U_\tau$ to $V_\tau$ passes through some vertex in $S$. 
\end{definition}

 In the previous analysis, the notion of vertex separator pops up as each edge needs to appear twice in the walk, and therefore, there is always some vertex in a given BlockStep that is making a middle appearance. However, this is no longer true for us as we may non-zero value walk that contain singleton edge and a priori, it is not at all clear whether vertex separator continues to control matrix norm bound in the regular setting.

To relate the block value to vertex separator, we start with a step-labeling of edges in the given block $\tau$. In fact, we restrict to step-labeling that does not involve singleton-step first. This would enable us to recover the usual graph matrix norm bounds. For concreteness, assume this is an $\tau$ block (as opposed to $\tau^\top$) and we are traversing from $U_\tau$ (left boundary) to $V_\tau$ (right boundary). Importantly the vertices in $U_\tau$ are already specified for us.
\paragraph{Building the Separator from BlockStep-Labels}
	For a block of $\tau$ traversing from $U_\tau$ and $V_\tau$, and let $L:E(\tau)\rightarrow \{F,R,H, \text{singleton}\}$ be a step-labeling of the given block, we build $S_B$ the separator for the given block as the following,\begin{enumerate}
		\item Include any vertex in $U_\tau\cap V_\tau$ into $S_B$;
		\item Include any vertex incident to $H$-step into $S_B$;
		\item Include any vertex incident to both $F,R$ edges into $S_B$;
		\item Include any vertex in $U_\tau$ incident to some $F/S$-step (note that $S$ is a surprise step as opposed to a $\text{singleton}$-step) into $S_B$;		\item Include any vertex in $V_\tau$ incident to some $R$-step into $S_B$;
		\item Additionally include any other vertex making middle appearance into $S_B$ in the given block.
	\end{enumerate}  

At this point, it is not clear whether $S_B$ built from the above process is a vertex separator for any block-labeling, not to mention whether it governs the matrix norm. In fact, it is possible in our case that $S_B$ is not a vertex separator for $\tau$ at all. That said, we first consider the case where there are no singleton-steps, and in this case, one can see the set $S_B$ built from above is indeed a vertex separator. This in part explains why matrix norm bounds are connected to vertex separator in the previous works for i.i.d. setting.

\begin{claim}
	In the absence of singleton-steps, $S_B$ is a vertex separator, i.e., any path from $U_\tau$ to $V_\tau$ passes through some vertex in $S_B$.
\end{claim}

\begin{proof}
To show $S_B$ is a vertex separator, it suffices for us to show any path from $U_\tau$ to $V_\tau$ passes through some vertex in $S_B$.
	Consider any such path, we first observe that if the path is trivial (i.e., no step is involved), such path is blocked by $U_\tau\cap V_\tau \subseteq S_B$. For the non-trivial path, such path cannot contain any $H$-step, as any $H$-step has both its endpoints included in $S_B$, and therefore blocked by the separator. Similarly, since each vertex incident to both $F$ and $R$-steps are included into the separator, any path from $U_\tau$ to $V_\tau$ must be of either all $F$ or all $R$-steps. 
	
	In the case of a path being of all $F$-steps, let $u$ be a vertex in which path intersects $U_\al$, this is a $U_\tau$ vertex incident to $F$-step and therefore included into the set $S_B$. Analogously, an all $R$-step patg passes through some vertex in $V_\tau$ incident to an $R$-step and included into $S_B$. Therefore, any path from $U_\tau$ to $V_\tau$ passes through the set $S_B$, and $S_B$ is indeed a vertex separator.
\end{proof}
From the above proof, we can also arrive at the following corollary immediately,
\begin{corollary}
	In the presence of singleton-steps, $S_B$ is not necessarily a vertex separator, however, any path not blocked by $S_B$ is all-singleton.
\end{corollary}

\begin{observation}
	Any vertex included into $S_B$ in the above process is a vertex making a middle appearance in the given block.
\end{observation}
With the connection between vertex separator and step-labeling established, we are now ready to bound the block-value of a given step-labeling, and see why it is indeed the norm bound given by vertex separator in the previous works. 

\subsection{Deducing Block-Value from Step-Labeling}
In this section, we combine the edge-value assignment scheme and our vertex-factor assignment scheme introduced earlier to bound the block-value of a given step-labeling
. On a high-level, we first consider the case where there is no singleton-edge, and show in this case we recover the usual graph matrix norm bounds. 

However, as we show in the last section, a step-labeling in the presence of singleton-edges is not necessarily a vertex separator, therefore it is not quite sufficient to consider step-labeling free of singleton-step. To handle such discrepancy, we show that any step-labeling $L$ with singleton-step, we can construct an alternative BlockStep-labeling $L'$ that is singleton-step free, and furthermore we have the BlockValue bound $B(L) \leq B(L')$. Therefore, we show that for the sake of maximizer of the block-value (and thereby of the matrix norm bound), it suffices for us to consider $L$ that is free of singleton-step, and for these step-labelings, we have a "genuine" vertex-separator as usual.

\begin{lemma}
	For a shape $\tau$, for any step-labeling $L$ of $\tau$ that does not contain singleton-step, let $S$ be its associated vertex separator built from the prescribed process,  we can bound its block-value by \begin{align*}
B(L)&\coloneqq \mathsf{vtxcost}(L) \cdot \edgeval(L) \\&\leq (\sqrt{n} \cdot q_\tau)^{|V(\tau))\setminus S| } \cdot \sqrt{\frac{1-p}{p}}^{|E(S)|} \cdot \sqrt{n}^{|I(\tau)|} \cdot q_\tau^{|V(S)\setminus U_\tau\cap V_\tau | } \cdot \max(1, q_\tau^{|E(\tau)\setminus E(S)| - |V(\tau)\setminus V(S)| }) 		 
	\end{align*} 
With some abuse of notation, we also write $B(S) \coloneqq B(L)$ defined as above as we observe the RHS quantity only has dependence on $L$ through $S$.
\end{lemma}

\begin{proof}
	The bound follows from observing the following,
	\begin{enumerate}
	\item By our vertex factor assignment scheme, any vertex outside the separator is making a first or last appearance in this block but not both since $L$ is free of singleton-step, and gets assigned a vertex cost of $\sqrt{nq_\tau} \leq \sqrt{n} \cdot q_\tau $;
	\item Any $H$ edge gets assigned an edge-value of $\sqrt{\frac{1}{p(1-p)}}\approx \sqrt{\frac{n}{d}}$, and any such edge is inside the separator, therefore, we have at most a factor of $\sqrt{\frac{n}{d}}^{|E(S)|}$ in total;
	\item Any excess edge gets assigned a factor of $q_\tau$ in its $F/R$-step, therefore we have a factor of $q_\tau^{|E(\tau) \setminus E(S)| - |V(\tau)\setminus S|}$ in total;
	\item Any vertex inside the separator other than $|U_\tau\cap V_\tau|$ is making a middle appearance and gives a factor of $q_\tau$, where we notice that a factor is spared for $|U_\tau\cap V_\tau|$ as no encoding is needed for such vertices.
\end{enumerate}
 
\end{proof}

\begin{lemma}
	For a shape $\tau$, for any of its step-labeling $L$ that contains singleton-step, if $\tau$ does not contain any floating component, there exists a separator $S$  whose corresponding block-value bound that at least matches the block-value bound of $L$. In other words, there exists vertex separator $S$ such that 
	\[ 
	B(L) \leq B(S) = (\sqrt{n}\cdot q_\tau)^{|V(L)\setminus S| } \cdot \sqrt{\frac{1-p}{p}}^{|E(S)|} \cdot q_\tau^{|V(S)\setminus U_\tau\cap V_\tau | } \cdot \max(1, q_\tau^{|E(\tau)\setminus E(S)| - |V(\tau)\setminus V(S)| })	 
	\]
\end{lemma}

\begin{proof}
	For starters, notice that our block-value bound factorizes over connected component, and it suffices for us to consider a single connected component.
	We first observe given a step-labeling, we can decompose $\tau$ into singleton-components by only considering singleton-step edges in $\tau$. 
%
	
	We now case on whether $S(L)$, the candidate vertex separator built from $L$ containing singleton-step, is a vertex-separator for $L$. If so, it suffices for us to verify our previous bound of  \[ 
	B(L) \leq (\sqrt{n}\cdot q_\tau)^{|V(L)\setminus S| } \cdot \sqrt{\frac{1-p}{p}}^{|E(S)|} \cdot q_\tau^{|V(S)\setminus U_\tau\cap V_\tau | } \cdot \max(1, q_\tau^{|E(\tau)\setminus E(S)| - |V(\tau)\setminus V(S)| })		
	\]
continues to hold for the given block when we take $S=S(L)$ even in the presence of singleton-step.
To see this, notice the previous argument bounding vertex factors in identical way for vertices outside the separator that are not making both first and last appearances,  and therefore, it suffices for us to focus on the "new" vertices making both first and last appearance due to the non-repetitiveness of singleton-edges. Any such vertex is a non-isolated vertex (otherwise if isolated, its double factor has already been accounted for in the previous proof by the factor of $\sqrt{n}^{|I(\tau) )|}$), and 
each such vertex gives an extra factor of $\sqrt{n}$. Call such vertex a \emph{flip vertex}, and \emph{unflipped} otherwise. That said, we can restrict our attention to these vertices and note that are all incident to some singleton edge that comes with an extra $\frac{1}{\sqrt{n}}$ decay. Therefore, the task is to assign a singleton edge for each such vertex that is making both first and last appearance in the block.

Consider the graph induced by singleton edges, we observe that any flipped vertex has a path of singleton edges to an unflipped vertex in the case there is no floating component since it has a path to $U_\al\cup V_\al$. Therefore we may consider a BFS to traverse flipped vertex from some unflipped vertex via singleton edges, and assign the singleton edge to the vertex it leads to, and observe that
\begin{enumerate}
	\item Each singleton edge, if non-excess, comes with a value $\frac{1}{\sqrt{n}}$ and gets assigned to the destination vertex, which is indeed the target we start off with;
	
	\item Otherwise, if an excess edge (i.e. a singleton edge that is also a surprise step leading to visited vertex), we incur $q^{O(1)}$ factor but it has an unassigned $\frac{1}{\sqrt{n}}$ factor as there must be at least another edge that is assigned to the destination vertex: if the vertex makes first appearance in the block, either it is not flipped and an extra $\sqrt{n}$ factor is not needed , or it is a flipped vertex but the extra $\sqrt{n}$ factor has been offset by the first singleton step that explores it\,,
\end{enumerate}
and this completes the proof for the case $S(L)$ remains as a separator for $\tau$.

On the other hand, if $S(L)$ is not a separator for $\tau$, it suffices for us to construct a separator $S$ and go through the above argument to show the desired block value bound of $B(L)\leq B(S)$. Towards this end, observe that there must be a path from $U_\tau$ to $V_\tau$ that does not intersect $S(L)$. By our construction of $S(L)$, any such path must be all singleton-step. Consider all such singleton-paths from $U_\tau$ to $V_\tau$ and let the union of all such paths be $T$. Let $X$ be the MVS for $T$ and now we consider $S\coloneqq S(L)\cup X$ as the new candidate vertex separator. 

Observe that $S=S(L)\cup X$ is indeed a vertex separator now as any path not blocked by $S(L)$ is in $T$ while any such path in $T$ from $U_\tau$ to $V_\tau$ is blocked by $X$. It now remains to verify the block S value of $S$ upper bounds the BlockValue of $L$.

Note that as the block-value bound factorizes over connected components, the factors of vertices/edges connected to $S(L)$ can be bounded by the previous argument, it suffices for us to focus on factors on the singleton path from $U_\tau$ to $V_\tau$.
Towards this end, we first consider a BFS from $X$ using the singleton-step in $L$ to  keep track of the change of factors in the BlockValue due to the addition of $X$.
\begin{enumerate}
	\item We start by observing for any vertex in $X$, it may either come from $U_\tau\Delta V_\tau$, or it is outside $U_\tau\cup V_\tau$. Any such vertex cannot be in $U_\tau\cap V_\tau$ since such vertices are automatically included into $S(L)$ already;
	\item Any vertex in $X\setminus (U_\tau \cup V_\tau)$ is making a first and last appearance in $L$, which gives a factor of $n$ . However, it now gives a factor of $\sqrt{n}$ in the bound on the RHS by being a vertex outside the separator, and this is a deficit of $\frac{1 }{\sqrt{n}}$;
	\item Any vertex in $X\cap (U_\tau \Delta V_\tau) $ is making either a first or last appearance in $L$ (but not both) which gives a factor of $\sqrt{n}$. Similarly to the above, it gives a factor of $1$ in the block-value bound of $B(S)$, which gives a deficit of  $ \frac{1}{\sqrt{n}} $;
	\item As we consider the BFS from $X$ using singleton-step, any such singleton-step corresponds to a gain of $\sqrt{n}$ as it gives value $1$ for block-value bound of $B(S)$while a factor of $\frac{1}{\sqrt{n}}$ to $B(L)$;
	\item Any vertex outside $U_\tau \cup V_\tau$ reached by a singleton-step gives a factor of $n$ in $B(L)$, while a factor of $\sqrt{n}$ to $B(S)$ Combining the change of vertex factor with the edge assigned to it, this is a change of \[ 
	\frac{\sqrt{n}}{n} \cdot \sqrt{n} =  1 \,;
	\]
	\item Any vertex in $U_\tau \Delta V_\tau$ reached by a singleton-step gives a change of $\sqrt{n}$ to both bounds, therefore there is no change in the vertex factor restricted to these vertices. Importantly, since it is still reached by a singleton-step, we pick up a change of $\sqrt{n}$ on these vertices;  
	\item Combining the above, notice for any vertex i$v \in X\setminus U_\tau \cup V_\tau$, we can assign at least one vertex in $U_\tau\setminus V_\tau$ and $V_\tau\setminus U_\tau$ reached by a BFS from $v$ without passing through any other vertex in $X$. As otherwise, $X\setminus \{v\}$ continues  being a separator of $T$ while this contradicts with the minimality assumption of $X$ to start with. With at least two vertices in $U_\tau\cup V_\tau$ assigned to each vertex $v$, we have a total change of \[ 
	 \frac{1}{n} \cdot (\sqrt{n})^2  =  1;
	\]
	\item Similarly, for any vertex $ v \in  X\cap (U_\tau\Delta V_\tau)$ (W.L.O.G. consider $v\in X\cap U_\tau $) ,we can assign at least one vertex in $V_\tau $ reached by the BFS via a singleton-step from $v$ without passing through other vertex in $X$, as again otherwise removing $v$ from $X$ preserves a separator for $T$ while contradicting with the minimality. This is a change of \[ 
	\frac{1}{\sqrt{n}} \cdot \sqrt{n}=  1 \,.
	\]
	 \end{enumerate}
	This proves that the block-value bound $B(S)\geq B(L)$ for $S= X\cup S(L)$ for step-labeling $L$	\end{proof}
	
	To wrap up, we have shown that any step-labeling has its block-value bounded by some separator, therefore, using a loose bound to crudely bound the number of step-labeling, the final bound follows as
		\begin{align*}
		B_q(\tau) \coloneqq \sum_{L : \text{ step-labelings for } E(\tau)} B(L) =  &\sum_{L : \text{ step-labelings for } E(\tau)} \mathsf{vtxcost}(L) \cdot \mathsf{edgeval}(L)\\
		&\leq c^{|E(\tau)|} \cdot \max_{S:\text{separator}} B_{\tau, q}(S)
	\end{align*}
	for some constant $c>0$, and we introduce the short-hand $B_{\tau, q}(S)$ to refer to $B(S)$ for shape $\tau$ when the dependence on shape $\tau$ may be unclear from the context.
	
	One may then observe that $c^{|E(\tau)|}  \cdot B_{\tau, q}(S)$ are indeed the norm bound given for graph matrix in the i.i.d. setting by \cite{JPRTX}, and in particular, the polynomial factors in $B_{\tau, q}(S)$ match exactly with the polynomial factors in the SMVS bound therein.  Since the $c^{|E(\tau)|}$ factor is dominated by the factors of $q$ (i.e. polylog factors in the $B_q$ bound) in the $B(S)$ bound, we have shown that these two bounds match up to lower-order dependence.
		
			 However, as introduced in our technical overview, such comparison no longer holds for the scalar example and potentially more. We next show that these are the only examples that a blow-up is incurred in the norm bounds. Formally, we show that a $\sqrt{n}$ blow-up is incurred for each tree-like floating component, i.e. edge-component not connected to neither of the left/right boundary.
			 
			 Before we get started, we first point out how floating component captures the scalar illustrated by the scalar random variable $p(x) = \sum_{i< j\in [n]} \chi_G(\{i,j\}) $ . This can be viewed as a shape with $U_\tau = V_\tau = \emptyset$ and $E(\tau) = p(x)$. Thus, with the boundary condition being empty, we indeed have a floating component and it is easy to see that is additionally tree-like.
			 
	We are now ready to augment the previous lemma to take into account of floating components.
	\begin{lemma}\label{lem:seplem}
		For any shape $\tau$, for any step-labeling $L$, there exits a separator $S$ whose corresponding block-value bound upper bounds that of $L$. In other words, there exists vertex separator $S$ such that for some constant $c>0$,
	\begin{align*}
		B(L) &\leq  (\sqrt{n}\cdot q_\tau)^{|V(\tau)\setminus S| } \cdot \sqrt{\frac{1-p}{p}}^{|E(S)|} \cdot \sqrt{n}^{|I(\tau)|} \cdot \sqrt{n}^{|float_{tree}(V(\tau)\setminus S) | } \\&\cdot \max(1, q_\tau^{|E(\tau)\setminus E(S)| - |V(\tau)\setminus V(S)| })  \cdot  q_\tau^{|V(S)\setminus U_\tau\cap V_\tau | } 	 \cdot c^{|E(\tau)|}
			\end{align*}
where $|float_{tree}(V(\tau)\setminus S)|$ is the number of tree-like floating components not connected to $S$.
	\end{lemma}
	
	\begin{proof}
	
	The key of this lemma is that each floating component gives at most a $\sqrt{n}$ blow-up.
		Again, we start with the observation that our block-value bound factorizes over connected components, and it suffices for us consider $C$ a floating component. In particular, it suffices for us focus on the components not connected to $S$ and those that are of all singleton-steps, as otherwise the previous proof applies: the component of vertices that make both first and last appearances are connected via singleton-step to vertices that are not making both appearances.
		
		It now suffices to bound the block-value restricted to all-singleton floating components, and its value is given by $ 
		n^{|V(C)|} \cdot \frac{1}{\sqrt{n}}^{|E(C)|}
		$ by our vertex-factor and edge-value assignment scheme. Notice in the case of tree-like floating components, we can bound it by \[ 
				n^{|V(C)|} \cdot \frac{1}{\sqrt{n}}^{|E(C)|}
 \leq \sqrt{n}^{|V(C)|} \cdot \sqrt{n} \leq  (\sqrt{n}\cdot q_\tau)^{|V(C)|} \cdot \sqrt{n} \
		\]
		and in the case of $|E(C)|\geq |V(C)|$, we can simply bound it by $(\sqrt{n} q_\tau)^{|V(C)|}$. Note that each bound on the RHS is the block-value factors in $B(S)$ associated with such floating component $C$, and this proves our desired lemma.
	\end{proof}

\subsection{Wrapping Up}
We are now ready to conclude by completing the proof to our main theorem.

\begin{proof}
Bby \cref{lem:seplem} and \cref{prop:sum-of-labelings}, we have \begin{align*}
	B(\tau) \leq \sum_L B(L) \leq c^{|E(\tau)|} \max_S B_\tau(S)
\end{align*}

	By design of our block-value function (and grouping the extra factor of $2$ from edge-value to the cost of identifying start of the walk), we have \[ 
	\E_{G\sim G_d(n)} \tr(M_\tau M_\tau^\top)^{q}  \leq O(n^{|U_\al|}) \cdot B_q(\tau)^{2q} \leq O(n^{|U_\al|})  B(\tau)^{2q}
	\]
	To obtain concentration of the matrix norm bound, we then appeal to the following proposition and this completes our result of matrix norm bounds for random regular graphs.
\end{proof}

\begin{proposition}
	For a given shape $\tau$, and a valid block-value bound $B_q(\tau)$, with probability at least $1-c^{-q/\log n } $,\[ 
	\|M_\tau\| \leq (1+o_d(1)) B_q(\tau)
	\]
\end{proposition}

\begin{proof}
	This follows from Markov's inequality, for any $\eps>0$, \begin{align*}
		\Pr[\|M_\tau\| \geq (1+\eps) B_q(\tau)] &\leq \Pr[\tr((M_\tau M_\tau^\top)^q) > (1+\eps)^{2q} B_q(\tau)^{2q})]\\
		&\leq \frac{\E[\tr((M_\tau M_\tau^\top)^{q})] }{(1+\eps)^{2q} B_q(\tau)^{2q}}\\
		&\leq O(n^{|U_\tau|/2q}) \cdot  (\frac{1}{1+\eps})^{2q}\\
		&\leq c^{-q/\log n}
	\end{align*}
	for any small enough constant $c>0$ and $q =\Omega(|U_\tau|\log n)$.
\end{proof}

%

\section{Sum-of-Squares Lower Bounds from Matrix Norm Bounds}
In this section, we show the same construction from i.i.d. setting continues to work in the regular setting. The main observation of our switching is the insight that as the matrix for i.i.d. case is largely based on moment method, and in the previous section of graph matrix norm bounds, we have shown the moments of these two distributions are largely the same (up to floating components), we can now plug in the new norm bounds to the existing analysis with some more care on the blow-up incurred by floating components.

\subsection{Our Results}

For our application we observe that it suffices for us to focus on shapes that do not have any floating component (in the dominant term) for the analysis of higher-degree Sum-of-Squares lower bounds, and therefore we can essentially plug in the new norm bound to the existing SoS analysis. To motivate our results, it is most helpful for us to recap the SoS formulation for the independent set problem and recap the
 previous result for Erdos-Renyi random graph. 

\begin{definition}[Pseudo-expectation of degree $\dsos$ ] For any $\dsos \in \N$, a degree-$\dsos$ pseudoexpectation in variables $x=(x_1,x_2, \dots, x_n)$ (denoted by $\pE$) is a linear map that that assigns a real number to every polynomial $f$ of degree $<d$ in x and satisfies: 1) normalization: $\pE[1] = 1$, and, 2) positivity: $\pE[f^2]\geq 0$ for every polynomial $f$ with degree at most $\dsos$ . For any polynomial $g(x)$, a pseudoexpectation satisfies a constraint $g = 0$ if $\pE[f\cdot g]=0$ for all polynomials $f$ of degree at most
$\dsos-deg(g)$.
	
\end{definition}

\begin{definition}[Axioms for independent set]
	Let $G$ be a n-vertex graph. The following axioms describe the $0-1$ indicators $x$ of independent sets in G:\begin{align*}
		&\forall i\in [n],   x_i^2 = x_i  \hspace*{12pt}  \text{(Booleanity);}\\
		&\forall (i,j)\in E(G), x_ix_j = 0 \hspace*{12pt} \text{(Independent set) }\,.
	\end{align*}

\end{definition}
Working with the above set-up, the existing result reads as the following,

\begin{theorem}[\cite{JPRTX}]\label{thm:ER-lb}
	With high probability over $G$ sampled from $G(n,\frac{d}{n})$, there is an absolute constant $c_0 \in \N$ such that for sufficiently large $n \in \N$ and $d \in [(\log n)^2, n^{0.5}]$,
and parameters $k$, $\dsos$ satisfying \[ 
k\leq \frac{n}{\sqrt{d}} \cdot \frac{1}{\dsos^{c_0} \cdot \log n}\,,
\]
there is a pseudo-expectation of degree $\dsos$ for independent set with objective value \[ 
\pE[\sum_{i\in [n]} x_i ] \geq (1-o(1)) k\,.
\]
\end{theorem}

We are now ready to introduce our main result of switching the input distributions, which states the main result of \cref{thm:ER-lb} holds with essentially the same parameter for inputs from $G_d(n)$ (up to lower order dependence of $\dsos$). 
\begin{theorem}[Switching to Lower bounds for $G_d(n)$]
		With high probability over $G$ sampled from the regular graph distribution $G_d(n)$, there is an absolute constant $c_1 \in \N$ such that for sufficiently large $n \in \N$ and $d \in [(\log n)^2, n^{0.5}]$, there is a pseudo-expectation of degree $\dsos<d^{1/10}$ for independent set with objective value \[ 
\pE[\sum_{i\in [n]} x_i ] \geq (1-o(1)) k\,.
\]
for
\[  k\leq \frac{n}{\sqrt{d}} \cdot \frac{1}{\dsos^{c_1} \log n}\,. \]
\end{theorem}

\begin{remark}
	
	Verbatim, our result for the regular graph distribution may seem to require an extra assumption of $\dsos<d^{1/10}$, however, we point out that produces makes minimal effect to the strength of our lower bounds in comparison the i.i.d. one as it is offset by the implicit $\poly(\dsos)$ dependence of the objective value. In other words, the extra assumption can be dropped taking a slightly larger constant $c_1>c_0$ for $c_0$ promised in the i.i.d. bound and set the objective value to be marginally smaller by $\poly\dsos$ factor.
	\end{remark}

\subsection{Technical Overview of Application to SoS Lower Bounds}
For higher-degere Sum-of-Squares lower bounds, the framework of "pseudo-calibration + PSD analysis" pioneered by the seminal work of \cite{BHKKMP16} has been by now a standard (if not the only known) route, and has been proven successful in several different contexts. Therefore, with its success in related problems, pseduo-calibration is a natural starting point for us. 

\paragraph{"Just Use Pseudo-calibration" Fails Again} The recipe of pseudo-calibration, and closely related to it the framework of low-degree polynomial hardness, calls for a pair of null and planted distribution, and an explicitly known orthogonal basis for null. 
To show that algorithm or analysis may be transported from the planted distribution to the null distribution, the low-degree polynomial framework show that it suffices to study the change of measure in the Fourier basis for the null-distribution basis. And this is the beginning for various low-degree polynomial analysis and higher-degree Sum-of-Squares lower bounds. However, this is also where one immediately runs into trouble. 
Unlike the prior scenario in \cite{JPRTX} where the difficulty is to find a quiet planted distribution, the regularity of our null distribution poses challenge to what an explicit Fourier basis would be for the regular graph distribution.   
 
 To address the absence of an explicitly known Fourier basis for the regular graph distribution, we settle with the closest known orthogonal basis we know, the $p$-biased Fourier basis for $G(n,\frac{d}{n})$. In fact, we make further compromise by recycling the moment matrix constructed for $G(n,\frac{d}{n})$, and plug in the input from random regular graph. In words, the candidate moment matrix construction in \cite{JPRTX} shows that there exists a matrix-valued function \[ \Lambda: \text{Graph} \rightarrow \text{SoS Moment Matrix of degree-}\dsos \] such that for typical graph sample $G$ comes from $G(n,\frac{d}{n})$, $\Lambda(G)$ is a valid pseudo-expectation of large objective value. However dangerous a move it may seem, we continue to apply the same matrix-valued function on graphs coming the regular-graph distribution, and observe that the independent-set constraints and large-value constraints continue to hold. We then defer further possible complications to the PSDness analysis.   
    
\paragraph{Applying the Approximate Factorization Machinery}
We briefly recap the overall idea of PSD analysis from approximate factorization here while we refer the interested reader to \cite{JPRTX, potechin2023machinery} for a more thorough treatment of this topic. In general, the machinery applies when the Fourier coefficient factorizes over edges and vertices, which also applies to our moment matrix construction  as we borrow the candidate matrix from the i.i.d. setting.
 
 To get started, the machinery starts by factorizing the target matrix as  $
	M = LQL^T 
	$
	for some graph matrix $L$ (left-shapes), $Q$ (middle-shapes) and $L^T$ (also viewed as right-shapes). The intention of the decomposition is that $\Lambda$ may have various eigenspaces corresponding to large eigenvalues and $L$ serves as projection matrix to each eigenspace of large eigenvalues. The PSDness then amounts to showing the middle matrix $Q$ is PSD. However, it turns out that the definition of $Q$ is less explicit due to intersection terms that arise from three-way product of $L, Q$ and $L^T$, and therefore, $Q$ needs to be defined in a recursive manner. Despite the technicality of constructing $Q$, specifically in further factorizing intersection terms, the up-shot of the decomposition remains the same that $Q$ is a matrix that does not contain large non-trivial eigenvalues. In particular, for the application to independent set, $Q$ may be seen as a perturbed identity matrix and showing its PSDness amounts to showing $Q-Id$, in particular, via a decomposition into a sum of graph matrices, having spectral norm $o(1)$. 

\paragraph{Handling the Discrepancy from Switching: Graph Matrix Norm Bounds for Regular Graphs}	
	The PSDness analysis of $Q$ is usually where moment method and spectral norm bounds comes into play as the proof strategy, and where the bulk of our switching work rests. To enable a smooth transfer via moment method, one natural idea is to show that if the low-degree moments do not differ much between these two distributions , and one may then establish similar spectral norm upper bounds for graph matrices for both input distribution by using polynomial approximation for the spectral norm upper bound. Since most of the prior analysis rely primarily on spectral norm upper bounds,  existing analysis with the new norm bounds would follow with ease if we can show comparable spectral norm upper bounds for input from the regular graph distribution.

	 A technical caveat of the above strategy is that not all the bounds are the same for these two distributions. In particular, low-degree polynomial can easily distinguish whether a graph comes from \Erdos-\Renyi or the regular distribution. 
	 In fact, even degree-$1$ polynomial suffices to distinguish such.
	 
	 \begin{examples}[Low-degree distinguisher for \Erdos-\Renyi and Regular Distribution]
	 	Consider the degree-$1$ of a single edge in the $p$-biased basis, $p(G) =\sum_{i,j,k\in [n]} \chi_G(\{i,j\})\chi_G(\{k,j\})$, one can bound w.h.p. $|p(G)|\leq \tilde{O}(\sqrt{n}^3)$ easily in \Erdos-\Renyi however, as a result of the non-orthogonality of the $p$-biased biases in the regular distribution, the scalar in regular distribution has typical magnitude $\tilde{\Omega}(\sqrt{n}^4)$ ignoring potential $\poly(-d)$ dependence.
	 \end{examples}
	 \begin{remark}
	 	We do not give a direct proof for this lower bound in $G_d(n)$, while we give a heuristic argument based on the distribution $G(n,\frac{d}{n})$ conditioned on vertices having degree exactly $d$.  It can be verified by direct moment calculation that a blow up of at least $O(\sqrt{n} \poly(-d))$.
	 		 \end{remark}
	 
	 Surprisingly, the existence of such low-degree polynomial distinguisher does not rule out the  above strategy of showing comparable spectral norm upper bounds for graph matrices for both distributions, and applying prior PSD analysis. Analogously to the well-known example of adjacency matrix in both distributions having roughly the same norm bound, we show that comparable norm bounds continue to for a large family of graph matrices, i.e., those of shapes that do not contain floating components. 
	  We defer the formal definitions to the subsequent section, but
	  in short, the polynomial deviation for the scalar bound is then reflected in our spectral norm upper bounds through the discussion on floating components.
	 
	For our application, the only tight case (in $\poly(n)$ order) comes from shapes that do not have floating components. On the other hand, as a result of our specially chosen truncation rule of connected truncation,
	floating components only appear at intersection terms where we show a more careful control of the coefficient allows us extra slack to tolerate the increase in the spectral norm upper bound compared with that of the i.i.d. setting.

\paragraph{Moment matrix construction}
To get started, we formally define our moment matrix as given by the following moments via the connected truncation for the vanilla pseudo-calibration. 

\begin{definition}[Moment matrix for independent set with connected truncation]
	We define the moments to be \[ 
	\pE[x_S] \coloneqq \sum_{\substack{ R\subseteq \binom{n}{2}\\ R \text{connected to S}\\ |V(R)\cup S|\leq D_V  } } (\frac{k}{n})^{|V(R)\cup S |} \cdot (-\sqrt{\frac{p}{1-p}})^{|E(R)| } \cdot \chi_R(G)
	\]
	for any $S\subseteq [n]$ s.t. $|S|\leq \dsos$.
We follow the notation of $\chi_R(G) = \prod_{e\in R} \chi_e(G)$, and $D_V \approx \dsos \log n $ is a truncation parameter to be specified later.
\end{definition}

\paragraph{Preliminaries for PSD Analysis}
We refer the reader for the prior work of \cite{JPRTX} for the full analysis while we strive to give a brief introduction to highlight the main lemmas that employ moment methods, and illustrate the necessary changes due to floating components. We start by recalling the following notations,

\begin{definition}[Moment matrix in the graph matrix basis]
	The moment matrix defined above as a linear operator can be equivalently written as \[ 
	\widetilde \Lambda = \sum_{\al:\text{connected}, |V(\al)|\leq D_V } \widetilde\lambda_\al \frac{M_\al}{|Aut(\al)|} 
	\]
	where we define \[ 
	\widetilde\lambda_\al \coloneqq \left(\frac{k}{n}\right)^{|V(\al)|}\cdot \sqrt{\frac{p}{1-p}}^{|E(\al)|}
	\]
	and $|Aut(\al)|$ is the size of automorphisnm group defined for shape $\al$ used to ensure each edge-set contributes once in the decomposition.
\end{definition}

For convenience, instead of working with $\Lambda$ directly, we usually work with a rescaled version of it obtained by left/right multiplying a diagonal matrix with entry being $\sqrt{\frac{n}{k}}^{|U_\al|}$, and we will primarily work with the rescaled coefficient defined as \[ 
\lambda_\al = \sqrt{\frac{n}{k}}^{|U_\al|+|V_\al|} = \left(\frac{k}{n}\right)^{|V(\al)|-|\frac{|U_\al|+|V_\al|}{2} } \cdot  (-\sqrt{\frac{p}{1-p}})^{|E(\al)| } 
\]

\begin{definition}[$\Pi$ independent-set indicator]
	Let $\Pi$ be a diagonal matrix for appropriate dimension for $\dsos$ with entries $\Pi[S, S] = 1[ S \text{ is an IS in  } G]$.
\end{definition}

At the backbone of the PSDness analysis is the notion of middle shape defined as. the following,
\begin{definition}[Middle shape]
	$\tau$ is a middle shape if $U_\tau$ and $V_\tau$ are the Minimum-Vertex-Separator of $\tau$, i.e., for any separator $S \subseteq V(\tau)$ that separates paths between $U_\tau$ and $V_\tau$, we have $|S|\geq |U_\tau|=|V_\tau|$.  
\end{definition}

Before we dig into the PSDness analysis, we first verify that this continues to give valid pseudo-moments when the underlying graph comes from random $d$-regular graph. We note that several proofs are in fact identical to the i.i.d. case while we re-verify them with the regular graph distribution for completeness.  

\subsection{Verification of Moment Constraints}
\begin{claim}[Satisfying independent set constraint]
	For any set $S$ with $|S|\leq \dsos $ , if $S$ is not an independent-set in $G$, we have \[ 
	\pE[x_S] = 0 
	\]
\end{claim}
\begin{proof}
	Note that for any set of edges $R$ outside $S$, we may group together the edges inside $S$, \begin{align*}
		\pE[x_S] & =  \sum_{\substack{ R\subseteq \binom{n}{2}\\ R: \text{connected to S}\\ |V(R)\cup S|\leq D_V  } } (\frac{k}{n})^{|V(R)\cup S |} \cdot (-\sqrt{\frac{p}{1-p}})^{|E(R)| } \cdot \chi_R(G)
\\
&= \sum_{\substack{ R\subseteq \binom{n}{2} \setminus \binom{S}{2} \\ R: \text{connected to S}\\ |V(R)\cup S|\leq D_V  } } (\frac{k}{n})^{|V(R)\cup S |} \cdot (-\sqrt{\frac{p}{1-p}})^{|E(R)| } \cdot \chi_R(G) \cdot \left(\sum_{H \subseteq \binom{S}{2}} (-\sqrt{\frac{p}{1-p}})^{|E(H)|  }\cdot \chi_H(G)    \right)\\
&=\sum_{\substack{ R\subseteq \binom{n}{2} \setminus \binom{S}{2} \\ R: \text{connected to S}\\ |V(R)\cup S|\leq D_V  } } (\frac{k}{n})^{|V(R)\cup S |} \cdot (-\sqrt{\frac{p}{1-p}})^{|E(R)| } \cdot \chi_R(G)  \cdot \frac{1}{(1-p)^{\binom{S}{2}}  } \cdot 1[S \text{is an ind-set in  } G]
	\end{align*}
	where we observe that subset of edges inside $S$ group to be a multiple of independent-set indicator of $S$.
\end{proof}
\begin{claim}[Large-value]
	W.h.p., the pseudo-moment matrix gives objective value $\sum_{i\in [n]} \pE[x_i] \geq   (1-o(1)) k$.
\end{claim}
\begin{proof}
	We highlight the observation that most calculations can be reduced to moment-calculations and our new norm bounds, in particular, graph matrix norm bounds for shapes $\tau$ with $U_\tau = V_\tau$. Notice this is equivalent to be bounding \begin{align*}
		\sum_{i\in [n]} \left|\pE[x_i  ] -\frac{k}{n}\right| = \sum_{i} \sum_{\substack{\al \subseteq \binom{n}{2} \\ E(\al) \neq \emptyset :  \text{ connected to } i \\ |V(\al)\cup \{i\}|\leq D_V }   } \left\| \widetilde{\lambda}_\al M_\al \right\|\\
	\end{align*}	
	for  $\al$ viewed as a shape with $U_\al = V_\al = \{i\}$ and $\tilde{\lambda}_\al = \left(\frac{k}{n}\right)^{|V(\al)| } \cdot (-\sqrt{\frac{p}{1-p}})^{|E(\al)|} $ the unscaled shape coefficient of $\al$. Crucially, we want to point out that in the calculation of the above, each shape $\al$ dose not containing any floating component, thus the norm bound applies in an identical way as that of the i.i.d. setting. We first switch the unscaled coefficient to the scaled coefficient of \[ 
	\lambda_\al = \sqrt{\frac{n}{k}}^{(|U_\al|+|V_\al|) } \cdot \widetilde{\lambda}_\al 
	\] 
	Since $U_\al = V_\al =\{i\}$, we have 
	\begin{align*}
		\sum_{i} \sum_{\substack{\al \subseteq \binom{n}{2} \\ E(\al) \neq \emptyset :  \text{ connected to } i \\ |V(\al)\cup \{i\}|\leq D_V }   } \left\| \tilde{\lambda}_\al M_\al \right\| \leq n  \cdot \frac{k}{n} \sum_{\al: |U_\al|=|V_\al|, \text{connected}, |V(\al)|\leq D_V, E(\al)\neq \emptyset}    \| \lambda_\al M_\al \|\ =o(k)
	\end{align*}
	where the last equation of $\cdot \frac{k}{n} \sum_{\al: |U_\al|=|V_\al|, \text{connected}, |V(\al)|\leq D_V, E(\al)\neq \emptyset }   \| \lambda_\al M_\al \|\ =o(1)$ is proven in our main lemma for PSDness concerning middle shapes in the subsequent section. 
		 \end{proof}
		 
\section{Acknowledgement} We would like to thank Pravesh Kothari, Madhur Tulsiani and Aaron Potechin for various discussions. We thank the anonymous reviewer from ITCS 2024 for pointing out that our particular SoS lower bound application in the regime of poly-logarithmic average-degree can also be obtained from Kim-Vu's sandwich conjecture, and all reviewers for their helpful suggestions for improving our writing.
\clearpage
\newpage
\bibliographystyle{alpha}
\bibliography{bib}
\clearpage
\newpage
\appendix

\section{Deferred PSDness Analysis with Swiwtched Norm Bounds}	
We now review the proof strategy in the previous works and highlight how our modified norm bounds may be plugged into the previous analysis with some technical twists to complete the proof for PSDness.

We adopt the notation and main lemmas from previous works of \cite{JPRTX, KPX24}, and we clarify each of the lemma in the subsequent section when we discuss them individually. However, we note here that these are the sole conditions to verify for our PSDness proof by the proof to the main theorem in \cite{JPRTX}.
 
\begin{lemma}[Main lemmas for PSDness (Lemmas 6.54 to 6.57 in \cite{JPRTX}]
	The following lemmas are true w.h.p. for $G\sim G_d(n)$,\begin{enumerate}
	
		\item (Non-trivial middle-shape is bounded) For all sparse permissible $\tau$ that is a middle shape such that $|V(\tau)|> \frac{|U_\tau|+|V_\tau|}{2}$ and $|E(\tau)|-|V(\tau)| \leq C \dsos$,\[ 
		\lambda_{\tau}' \cdot \frac{\|M_\tau\|}{|Aut(\tau)|}\leq \frac{1}{c(\tau)}
		\]
		  for some defined slack function $c(\tau)$ in \cite{JPRTX} to be specified in later section,
		 		\item (Intersection term is bounded)  
		 \label{lem:intersection-terms}
    For all $j \geq 1$ and sparse permissible $\gam_j, \dots, \tau, \dots, \gam_j'$ such that for each shape $|E_{mid}(\al)| - |V(\al)| \leq C\dsos$,
    \[ \sum_{\substack{\text{nonequivalent} \\ P \in P^{mid}_{\gam_j,\dots,\gam_j'}}} N_{P}(\tau_P) \lambda'_{\gam_j \circ \cdots \circ \gam_j'^\top} \frac{\norm{M_{\tau_P}}}{|Aut(\tau_P)|} \leq \frac{1}{c(\tau) \prod_{i=1}^j c(\gam_i)c(\gam_i')}.\]
		
		\item (Truncation error)  \[ 
\text{truncation error} \preceq n^{-\Omega(C\dsos)} \pi.		\]
		\item (Well-conditionedness) (Sum of left shapes is well-conditioned)
    \label{lem:left-right-sep-conditioned}
    \[\left(\sum_{\substack{\text{sparse,}\\\text{permissible} \\ \sigma \in L}}\lambda'_\sigma \frac{M_\sigma}{|Aut(\sigma)|}\right)\left(\sum_{\substack{\text{sparse,}\\\text{permissible} \\ \sigma \in L}}\lambda'_\sigma \frac{M_\sigma}{|Aut(\sigma)|}\right)^\top
    \succeq n^{-O(\dsos)} \pi\]	\end{enumerate}
\end{lemma}

We now proceed to verify key claims in the above conditions, and we leave the full verification to later versions of this paper.

\paragraph{Verification of Non-trivial Middle Shapes} Let us first unpack the lemma statement for the reader. The shapes of focus here are the middle shapes: shapes such that $U_\tau, V_\tau$ are both MVS of the shape, and additionally the shape must contain some vertex other than $U_\tau\cap V_\tau$, hence the vertex constraint. Additionally, it restricts its scope further to shapes that satisfy the given sparsity bound, as in a sparse random graph, shapes containing too many edges are likely forbidden and can be bounded crudely.

To make the statement more readable, we advise the reader to first consider a single middle shape and verify \[
\lambda_\tau \frac{\|M_\tau\|}{|Aut(\tau)|} \leq o(1) \numberthis \label{eq:midshapecharging}
 \]
and note that this is the heart of middle-shape analysis while the full statement incorporates it with an explicit slack function that is later used to extend the above bound on a single shape to the sum of a collection of shapes so that \[ 
\sum_\tau \lambda_\tau \frac{\|M_\tau\|}{|Aut(\tau)|} \leq o(1)
\]
\begin{remark}
	The use of coefficient $\lambda'_\tau$ is a result of the technical analysis as we restrict our analysis to sparse shapes, and may be viewed as $\lambda'_\tau = (1+o(1)) \lambda_\tau$ for readability. 
\end{remark}

We defer the full verification with slack function to the full version of the paper, while we highlight how our norm bound applies here for the above bound on a single shape. The extension to the slack function is completely mechanical and follow verbatim from \cite{JPRTX}.

 To see the switch is possible here,  we crucially observe that there is no floating component in the definition of middle shape in  \cite{JPRTX} as a result of connected truncation, and therefore the argument for middle shape if fact follows immediately once we plug in the new norm bound. Ignoring subpolynomial order dependence in the norm bound that is handled via the slack function, \cref{eq:midshapecharging} becomes \[ 
 \left(\frac{k}{n}\right)^{|V(\tau)|-\frac{|U_\tau|+|V_\tau|}{2} } \cdot (\sqrt{\frac{1-p}{p}})^{|E(\tau)|} \cdot \tilde{O}( \max_{S:\text{separator}}  \sqrt{n}^{|V(\tau)\setminus S| } \cdot \sqrt{\frac{1-p}{p}}^{|E(S)| }  )
 \]
 since there is no floating component or isolated vertex. To see that the above is at most $o(1)$, observe that \begin{enumerate}
 	\item $S$ is a separator for $\tau$, and by construction, $U_\tau$ and $V_\tau$ are both MVS of $\tau$, therefore $|S|\geq \frac{|U_\tau|+|V_\tau|}{2}$, therefore, we can assign a coefficient of $\frac{k}{n}$ for each vertex outside the separator;
 	\item Each vertex is connected to $U_\tau$ and $V_\tau$, and therefore the separator $S$ by middle shape assumption, we can consider a BFS from $S$ to traverse vertices outside the separator, and assign each vertex the edge (which comes with a coefficient of $\sqrt{\frac{p}{1-p}}$) that explores it in the process;
 	\item Each vertex contributes a factor of $\tilde{O}(\sqrt{n})$ and gets assigned a coefficient of $\frac{k}{n}\cdot \sqrt{\frac{p}{1-p}}$, and this is at most $o(1)$ by our choice of $k$.
 	\item Edges inside the separator give $\sqrt{\frac{1-p}{p}}$ to the norm bound while it also comes with a coefficient of $\sqrt{\frac{p}{1-p}}$, which offsets each other.
 \end{enumerate}

\paragraph{Verification with slack function}

We  first recall the choice of slack function $c(\tau)$ in \cite{JPRTX} is chosen to be \[ 
c(\tau) \leq 40(2C\dsos^{4C'+2})^{|V(\alpha)| - \frac{|U_{\alpha}| + |V_{\alpha}|}{2}}
\]
For our application, we note that it suffices for us by pick the same function except by potentially choosing a slightly larger constant $C'$. 

Therefore, both of our spectral norm bounds match up to the dependence of $c^{|E(\tau)|}$ in our norm bound. However, this is straightforwardly offset by the sparsity of our shape and the slack function by picking a larger constant. That said, the proof for non-trivial middle shape lemma follows from that of \cite{JPRTX} and by noting the comparable upper bound for these two distributions.

\paragraph{Verification of Intersection Term Analysis}
We follow the same strategy in middle shape and give a rough analysis by ignoring the precise slack function. The major difference here is that it is no longer true we can ignore floating component and its corresponding $\sqrt{n}$ blow-up as in middle shape. However, such blow-up is offset by the gap promised in the analysis in \cite{JPRTX, KPX24}.

Before we highlight this difference, let us start by some formal definition of intersection shape.

\begin{definition}[Middle intersection]\label{def:middle-intersection}
    Let $\gamma, \gamma'$ be left shapes and $\tau$ be a shape such
    that $\gamma \circ \tau \circ \gamma'^\top$ are composable.
    We say that an intersection pattern $P \in P_{\gam, \tau, \gam'^\top}$
    is a \emph{middle intersection} if
    $U_{\gamma}$ is a minimum vertex separator in $\gamma$ of $U_\gam$ and $V_\gam \cup Int(P)$. Similarly, $U_{\gam'}$ is a minimum vertex separator
    in $\gam'$ of $U_{\gam'}$ and $V_{\gam'} \cup Int(P)$.
    Finally, we also require that $P$ has at least one intersection.
    
    Let $P^{mid}_{\gam,\tau,\gam'}$ denote the set of middle intersections.
\end{definition}
\begin{remark}
For middle intersections we use the notation $\tau_P$ to denote the resulting 
shape, as compared to $\al_P$ which is used for an arbitrary intersection pattern.
\end{remark}
At this point, we want to point out that floating component arises for the same reason isolated vertex arises in the previous analysis: vertex intersection may cause an edge to appear with multiplicity $>1$ whose action on the matrix norm bound is equivalent to replacing it with a scalar $1$ that effectively removes the edge. As edges get removed, we may then have vertices being isolated and components being disconnected from the left and right boundary. 

Now we show that middle intersections have small norm. We focus on the first level of intersection terms; the general case of middle intersections then follows by induction. 

\begin{proposition}\label{prop:informal-intersection-terms}
    For left shapes $\gam, \gam'$, proper middle shape $\tau$ such that every
    vertex in $\tau$ is connected to $U_\tau \cup V_\tau$, and a middle
    intersection $P \in P^{mid}_{\gam, \tau, \gam'}$,
    $ \lambda_{\gam \circ \tau \circ \gam'^\top} \|M_{\tau_P}\| =o(1)$.
\end{proposition}

\begin{proof}
	Let us unpack here,  \begin{align*}
		&\left(\frac{k}{n}\right)^{|V(\gam\circ\tau\circ\gp)|- \frac{|U_\gam|+|V_\gp|}{2} }\cdot  \left(\sqrt{\frac{p}{1-p}} \right)^{|E(\gam\circ\tau\circ\gp)|}\cdot \max_{\substack{(\psi, S)\\ \psi: \text{linearization of $\tau_P$}\\ S:\text{a separator for }\psi }}  \tilde{O}( \left(\sqrt{\frac{1-p}{p}} \right)^{\theta(\psi)}  \left(\sqrt{\frac{1-p}{p}}\right)^{E_{\psi}(S)}\sqrt{n}^{V(\tau_P)\setminus V(S) }\\&\cdot  \sqrt{n}^{I_{\psi}}\cdot  \float(\psi) ) 
\leq o(1)
	\end{align*}
where we define a linearization of $\tau_P$ as a mapping from multi-edge to $E(\tau_P)$ to either $0/1$, and we write $\psi(\tau_P)$ as the shape obtained by removing the multi-edge that gets $0$ in the linearization (edge-removal), $I(\psi)$ the set of isolated vertices under $\psi(\tau_P)$   and $\float(\psi)$ factor on the floating component on $\psi(\tau_P)$ given by our norm bounds on random d-regular graph.

We follow our middle-shape strategy above, it suffices for us to assign a factor of $\frac{k}{n}$ along with an edge to each vertex outside the separator as we have \[ 
\frac{k}{n} \cdot \sqrt{\frac{p}{1-p}} \cdot \tilde{O}(\sqrt{n}) < o(1)
\] and an additional pair if the vertex is isolated or floating (as each incurs an extra factor of $\sqrt{n}$). 
The factor assignment of $\frac{k}{n}$ follows from the application of intersection tradeoff lemma in Lemma 6.8 of \cite{JPRTX}, and it suffices for us to focus on the edge factor assignment to handle to extra $\sqrt{n}$ for floating component. In particular, it suffices for us to identify an extra edge that vanishes for each floating component.

We start by observing that each edge that vanishes in $\psi$ must be a multi-edge to start with in $\tau_P$, and gives $2$ factors of the edge-coefficient. Therefore, it suffices for us to show \[ 
|E_\psi|+ \phantom(\psi) \geq |V(\tau_P)\setminus V(S)|+|E_\psi(S)|+|I_\psi| +|\float_\psi|  \numberthis
\]
where we define $\phantom(\psi)$ to be the multiplicity of edges that vanish in $\psi$.

 We start by considering $\psi$ the linearized graph (recall $\psi$ may have phantom edges removed) and put phantom edges back along the way. Take $S$ to be the SMVS (i.e. the norm bound maximizer) for $\psi$, we then consider the following recursive process to traverse $\psi$ via edges in $E(\tau_P)$.
	
	 Throughout the process, vertices in the graph $V(\psi)$ can be partitioned as the following,
	\begin{enumerate}
		\item $V_W\subseteq V(\tau_P)$: those reachable from $W$ via edges in $E_\psi$ or already in $W$;
		\item For a connected component $C\subseteq V(\psi)\setminus V_W$ in $E_\psi$ while not yet reachable from $W$, it is either a non-floating component, or a floating component;
		\item For isolated vertices in $V(\psi)\setminus V_W$, we can group them according to the phantom edges into components connected by phantom edges.
	\end{enumerate} 
	We then consider the following process to assign phantom edges (i.e. those that become from $mul\geq 2$ to $0$ in $\psi$), \begin{enumerate}
		\item Let $W$ be the current set of vertices visited (initialized to be $S$ the  SMVS of $\psi$);
		\item Explore the vertices (not yet in $W$) while connected to $W$ via edges in $E(\psi)$, i.e., assign the edge to each vertex it leads to;
		\item Explore a component connected to $W$ via some phantom edge;
		\item For vertices outside $W$ and not reachable via phantom from $W$, there must be a phantom-edge connecting two different components, process that phantom edge and explore both components;
		\item Repeat this process until all vertices are pushed into $W$.
	\end{enumerate}
	Next we proceed to show this achieves our assignment goal, in particular, it assigns a phantom edge to a floating component. The proof on factors outside floating components is given in proof to Lemma 6.9 in \cite{JPRTX} while we highlight the extra factor for the floating component here: this is analogous to our charging argument handling the effect of singleton edges on the block-value bound in the norm bound discussion, each such component has a path to $U \cup V$ via phantom edges, and we may then consider the traversal process to assign phantom edges. Note that each such phantom edge comes with $2$ multiplies and may have $1$ assigned to the component for the original $\sqrt{n}$ factor of the first vertex of the component, however, we observe that the second multiplicity remains unassigned in the above process and can be assigned to capture the extra blow-up for the floating component, and this completes our bounds for intersection terms
	
The calculation with the slack function then follows immediately as that of middle-shape. Similarly, the proof for truncation terms and well-conditionedness are both straightforward applications of moment methods with our new norm bounds: in particular, there is no floating component in discussion so it suffices for us to use the comparable spectral norm upper bounds. Therefore we have shown the main lemmas continue to hold for $G_d(n)$, and this completes our Sum-of-Squares lower bounds.
\end{proof}

\end{document}

%% file: math_commands.tex
\def \N {\mathbb{N}}

\def \E {\mathbb{E}}

\def \eps {\epsilon}
\def \al {\alpha}
\renewcommand{\Pr}{\mathop{\bf Pr\/}}
\newcommand{\wt}[1]{\widetilde{#1}}

\newcommand{\Var}{\mathrm{Var}}

 \def\1{\bm{1}}

\newcommand{\calF}{\mathcal{F}}

\newcommand{\pE}{\widetilde{\E}}

\usepackage{amsmath,amssymb,amsthm,amsfonts,latexsym,bm,bbm,xspace,graphicx,float,mathtools,mathdots,physics}
\usepackage{braket,caption,subcaption,ellipsis,xcolor,textcomp,hhline,pifont,combelow,booktabs}
\usepackage[colorlinks=true, allcolors=blue]{hyperref}
\usepackage{color}
\usepackage{times}
\usepackage{fullpage}
\usepackage{tikz-cd}
\usepackage[shortlabels]{enumitem}
\usepackage{thm-restate}
\usepackage{cleveref}
\usepackage{mdframed}

\newcommand{\NP}{\mathsf{NP}}

\newtheorem{theorem}{Theorem}[section]
\newtheorem{lemma}[theorem]{Lemma}
\newtheorem{claim}[theorem]{Claim}
\newtheorem{proposition}[theorem]{Proposition}

\newtheorem{corollary}[theorem]{Corollary}

\newtheorem{definition}[theorem]{Definition}
\newtheorem{remark}[theorem]{Remark}

\newtheorem{observation}[theorem]{Observation}

\newtheorem{examples}[theorem]{Example}

\newcommand{\poly}{\operatorname{poly}}
\newcommand{\polylog}{\operatorname{polylog}}

\DeclareMathAlphabet{\mathsfit}{\encodingdefault}{\sfdefault}{m}{sl}
\SetMathAlphabet{\mathsfit}{bold}{\encodingdefault}{\sfdefault}{bx}{n}

\renewcommand{\wt}{\mathsf{weight}}

\mathchardef\mhyphen="2D

\newcommand{\calE}{\mathcal{E}}

\newcommand\numberthis{\addtocounter{equation}{1}\tag{\theequation}}

\newcommand{\dsos}{\mathsf{dsos}}
\newcommand{\mul}{\text{mul}}
\newcommand{\ind}[1]{1_{{#1}}}

\newcommand{\float}{\mathsf{float}}

\newcommand{\gam}{\gamma}

\renewcommand{\int }{\text{Int}}

\renewcommand{\phantom}{\text{phantom}}

\newcommand{\cnorm}{c_{\mathsf{norm}}}

\newcommand{\gp}{{\gamma'}}

\newcommand{\edgeval}{\mathsf{edgeval}}

\newcommand{\Erdos}{Erd\H{o}s\xspace}
\newcommand{\Renyi}{R\'enyi\xspace}
